\documentclass{article}
\usepackage{setspace}
\usepackage{geometry}
\usepackage{amsmath}
\usepackage{amsthm}
\usepackage{amssymb}
\usepackage[all,pdf]{xy}
\usepackage{hyperref}
\usepackage{indentfirst}
\usepackage{cite}
\usepackage{graphicx}
\newcommand{\lt}{<}
\newcommand{\gt}{>}
\hypersetup{hidelinks}
\newtheorem{theorem}{Theorem}
\newtheorem{lemma}{Lemma}
\newtheorem{remark}{Remark}
\newtheorem{prop}{Proposition}

\newtheorem{assumption}{Assumption}
\title{Variance Reduction for Causal Inference}
\author{
 Kangjie Zhou\thanks{Department of Statistics, Stanford University} 
 \and 
 Jinzhu Jia\thanks{Department of Biostatistics, School of Public Health, Peking University}
}
\date{}
\begin{document}
	\maketitle
	
	\begin{abstract}
		Propensity score methods have been shown to be powerful in obtaining efficient estimators of average treatment effect (ATE) from observational data, especially under the existence of confounding factors. When estimating, deciding which type of covariates need to be included in the propensity score function is important, since incorporating some unnecessary covariates may amplify both bias and variance of estimators of ATE. In this paper, we show that including additional instrumental variables that satisfy the exclusion restriction for outcome will do harm to the statistical efficiency. Also, we prove that, controlling for covariates that appear as outcome predictors, i.e. predict the outcomes and are irrelevant to the exposures, can help reduce the asymptotic variance of ATE estimation. We also note that, efficiently estimating the ATE by non-parametric or semi-parametric methods require the estimated propensity score function, as described in Hirano et al. (2003)\cite{Hirano2003}. Such estimation procedure usually asks for many regularity conditions, Rothe (2016)\cite{Rothe2016} also illustrated this point and proposed a known propensity score (KPS) estimator that requires mild regularity conditions and is still fully efficient. In addition, we introduce a linearly modified (LM) estimator that is nearly efficient in most general settings and need not estimation of the propensity score function, hence convenient to calculate. The construction of this estimator borrows idea from the interaction estimator of Lin (2013)\cite{Lin2013}, in which regression adjustment with interaction terms are applied to deal with data arising from a completely randomized experiment. As its name suggests, the LM estimator can be viewed as a linear modification on the IPW estimator using known propensity scores. We will also investigate its statistical properties both analytically and numerically.
	\end{abstract}
	
	\section{Introduction}
	\noindent Methodological advances have resulted in great progress on estimation of average treatment effect (ATE) in observational studies, especially in situations where there are a large number of covariates, including instrumental variables, confounding factors and outcome predictors. A great amount of techniques dedicating to reduce this huge dimensionality have been discovered, especially the propensity score matching (PSM), see Rosenbaum and Rubin (1983)\cite{RosenbaumRubin1983} for a complete view. The validity of propensity score matching is ensured by the so-called unconfoundedness assumption, i.e. there are no unmeasured confounders that are not included in the covariates. Based on the unconfoundedness assumption (or ignorable treatment assignment assumption), several $\sqrt{n}-\text{consistent}$ estimators of ATE were proposed, among which the most popular one makes use of propensity score as an inverse weight to balance the covariates' distribution among treatment group and control group. The readers can refer to Imbens (2004)\cite{Imbens2004} for a complete review of these estimators and related approaches. While using propensity score matching to estimate ATE, which covariates should be included is an important but confusing issue.
	
	Previously a "throw in the kitchen sink" mentality has been used to include as many covariates in to the propensity score model as one can (Shortreed, 2017)\cite{Shortreed2017} in pursuit of fulfillment of unconfoundedness assumption. But one may sometimes include too many instrumental variables that may cause efficiency loss of ATE estimators. On the one hand, according to previous studies (Lin, 2013)\cite{Lin2013}, in a completely randomized experiment, incorporating pre-treatment characteristics that are not effected by the treatment assignment will lead to efficiency gain, when compared to the simple intention-to-treat (ITT) estimator. Specifically speaking, incorporating interaction terms between such covariates and treatment assignment into linear regression help decrease asymptotic variance of regression coefficient of exposure, i.e., the interaction estimator of ATE. Simulation studies from Brookhart et al. (2006)\cite{Brookhart2006} also supported the similar result in observational studies, when the treatment design may not be randomized due to the existence of confounders. We will prove this result rigorously for several estimators in later sections. 
	
	But on the other hand, including variables associated with exposure but not potential outcomes, i.e., instrumental variables will enlarge asymptotic variance of estimator of ATE. This idea has been suggested and examined numerically by several researchers. De Luna et al. (2011)\cite{DeLuna2011} addressed the variance inflation effect caused by including additional instrumental variables and attached great significance to reducing the dimension of covariates by variable selection. They proposed a covariate selection method in order to avoid efficiency loss and eliminate bias. Patrick et al. (2011)\cite{Patrick2011} also used an empirical illustration to highlight this idea. Hahn (2004)\cite{Hahn2004} proved that including some specific kind of instrumental variables, i.e., those are independent of potential outcomes when conditioning on confounders will inflate the semi-parametric efficiency bound, and we will step further to show that this result holds for more estimators.

	Now we turn to consider efficient estimation of ATE under the unconfoundedness assumption. Usually, a semi-parametric $\sqrt{n}$-consistent estimator of ATE is said to be asymptotic efficient if its asympotic variance achieves the semi-parametric efficiency bound. Newey (1990)\cite{Newey1990} formally proposed the idea of semi-parametric efficiency bound. For a semi-parametric model, its efficiency bound is defined by the infimum over asymptotic variances of all possible semi-parametric estimators. It is in fact a semi-parametric generalization of Cramer-Rao lower bound (Cramer et al., 1946\cite{cramer2016mathematical}; Rao, 1945\cite{rao1992information}) for parametric models developed in mathematical statistics. The conception of semi-parametric lower bound can be extended to a wide range of statistical models, including the Neyman-Rubin Causal Model. Hahn (1998)\cite{Hahn1998} also investigated the imputation estimator using estimation of both propensity score function and conditional expectations of potential outcomes. An astonishing conclusion is that knowing propensity score is not necessary and may even be harmful for efficient estimation of ATE sometimes, for instance, the IPW estimator. Assume that the propensity scores are known, then the IPW estimator of ATE using known propensity scores (called the original IPW estimator in sequel) is not efficient in the meaning that its asymptotic variance is strictly larger than the semi-parametric efficiency bound. However, another version of the IPW estimator in which propensity scores are replaced by their $\sqrt{n}-$consistent estimators, is proven to be asymptotically efficient. The same phenomenon also happens for the imputation estimator, i.e., if the estimators of propensity scores are replaced by their true values, asymptotic variance will inflate. But in the mean time, we must notice that, although using nonparametric estimators of propensity score function help improve statistical efficiency, the estimation procedure itself can be rather difficult and erroneous and requires a lot of regularity conditions that are usually hard to satisfy. When the proposity scores are known, it is necessary to propose an estimator that uses the known propensity score function and possess decent efficiency.
	
	Assume that the propensity score function is known. Accordingly, we will propose a new estimator of ATE, namely the linearly modified (LM) estimator which can be viewed as a linear modification of the original IPW estimator to address both issues aforementioned. It utilizes difference between inverse probability weighted mean of covariates in the treatment group and the control group as a correction term to modify the original IPW estimator. The LM estimator only depends on known propensity score function and observed outcomes, hence easy to calculate. Compared to the original IPW estimator, the LM estimator is far more efficient in the sense that its asymptotic variance is very close to the semi-parametric efficiency bound under most general settings. This claim will be verified by simulations when the potential outcomes satisfy a linear model and the propensity scores satisfy a logistic model. Compared to statistically efficient estimators mentioned above, it does not require consistent estimation of propensity scores, hence is easier to calculate and asks for less regularity conditions.
	
	The remainder of this paper is organized as follows. Section 2 provides necessary prerequisites about basic knowledge and notation of causal inference. In Section 3, we discuss estimators of ATE constructed by propensity score matching and analyze their statistical properties. We will then prove that including outcome predictors always help improve statistical efficiency of estimators mentioned above. We also review variance bound of semi-parametric estimators and illustrate the variance inflation effect of introducing a specific kind of instrumental variables rigorously, Hahn (2004)\cite{Hahn2004}'s result will be extended. In section 4, we will propose our linearly modified estimator and discuss its properties. Comparison with other estimators is also provided. We will present some numerical simulations in Section 5 to illustrate that the LM estimator is nearly efficient in general situations. We close with a conclusion discussion in Section 6. Technical Proofs will be provided in Appendix.
	
	\section{Preliminaries}
	
	\subsection{Causal Inference}
	
	\noindent The major goal of causal inference study is to estimate the average treatment effect (ATE) from observational data or randomized experiments. In observational study, embracing the Neyman-Rubin model,  the ATE is defined to be the difference between expectations of two potential outcomes, denoted by $Y^T$ and $Y^C$ respectively. Let $Y$ denote the observed outcome, if the corresponding individual is assigned to the treatment group, then $Y=Y^T$, if the corresponding individual is assigned to the control group, then $Y=Y^C$. Let $D$ denote the treatment assignment, which means that $D=1$ if the corresponding individual is assigned to the treatment group, $D=0$ if the corresponding individual is assigned to the control group. Therefore, $Y=DY^T+(1-D)Y^C$, and we define ATE by
	\begin{equation}
	\text{ATE} = \mathbf{E} Y^T - \mathbf{E} Y^C.       \tag*{}
	\end{equation}
	The basic problem of causal inference is that we can never know $Y^T$ and $Y^C$ for every individual simultaneously, because either of them is observed but never both. Hence we need to speculate on the unobserved potential outcome in order to obtain consistent and efficient estimators of ATE. (Rosenbaum and Rubin, 1983)\cite{RosenbaumRubin1983}
	
	Generally speaking, the treatment design, denoted by $D$, is not randomized and therefore the probabilities $\mathbf{P}(D=1)$ and $\mathbf{P}(D=0)$ may vary among different subgroups of the population. If $D$ is completely randomized, then under Neyman-Rubin model (Neyman, 1923\cite{Neyman1923}; Rubin, 1974\cite{Rubin1974}; Holland, 1986\cite{Holland1986}), we know that $(Y^T, Y^C)$ is unconditionally independent of $D$. Hence we have
	\begin{equation}
	\mathbf{E} Y^T - \mathbf{E} Y^C = \mathbf{E} \lbrack Y^T \vert D=1 \rbrack - \mathbf{E} \lbrack Y^C \vert D=0 \rbrack,      \tag*{}
	\end{equation}
	based on which a large number of consistent estimators of ATE can be then developed and their statistical properties can be examined. The above equality also holds true if covariate vector, say $X$ in the model only contains outcome predictors, i.e., a set of pre-treatment characteristics that are correlated to potential outcomes $Y^T$ and $Y^C$ but independent of treatment assignment $D$. In observational studies, when this assumption is violated, we must assume the unconfoundedness assumption instead in order to reduce bias to obtain consistent estimators of ATE.
	
	\subsection{Assumptions}
	\noindent In most of observational studies, $D$ is not randomized and therefore $(Y^T, Y^C)$ and $D$ are not unconditionally independent, but independent when conditioning on a set of covariates, denoted by $X$ ($X$ may be a random vector), i.e.
	\begin{equation}
	(Y^T, Y^C) \perp D \vert X.          \tag*{}
	\end{equation}
	Here "$\perp$" means independence. This assumption is commonly called the "Unconfoundedness Assumption", which means that there is no unmeasured confounder, i.e., all possible confounders that influence both treatment assignment and potential outcomes are contained in $X$. In Rosenbaum and Rubin (1983)\cite{RosenbaumRubin1983}, a function $b(X)$ of $X$ is defined as a balancing score if $(Y^T, Y^C) \perp D \vert b(X)$. Hence the most trivial balancing score is $X$ itself, and we would like to introduce another balancing score, known as the coarsest function of $X$ to become a balancing score (Rosenbaum and Rubin, 1983)\cite{RosenbaumRubin1983}, i.e., the propensity score $p(X)$.
	
	The propensity score is defined as the probability of an individual to receive treatment conditional on the covariates, i.e.,
	\begin{equation}
	p(X) = \mathbf{P} (D=1 \vert X) = \mathbf{E} \lbrack D \vert X \rbrack,   \tag*{}
	\end{equation}
	and we make the following Overlap Assumption on the propensity scores:
	\begin{equation}
	0 \lt p(X) \lt 1, a.s. \ \text{for} \ X.        \tag*{}
	\end{equation}
	If both Unconfoundedness Assumption and Overlap Assumption are satisfied, then generally we say that treatment assignment is strongly ignorable. Refer to Rosenbaum and Rubin (1983) for more details about this definition.
	\begin{assumption}
		Unconfoundedness Assumption: $(Y^T, Y^C) \perp D \vert X$.
	\end{assumption}
	
	\begin{assumption}
		Overlap Assumption: $0 \lt p(X) \lt 1, a.s.\ \text{for} \ X$.
	\end{assumption}
	\begin{remark}
		The Overlap assumption ensures that the treatment design is non-degenerate, i.e., both the treatment group and the control group is not empty. Some literature impose a stronger version of the Overlap assumption as: for some $\epsilon \gt 0$,
		\begin{equation}
		\epsilon \lt p(X) \lt 1-\epsilon, a.s.\ \text{for} \ X.      \tag*{}
		\end{equation}
	\end{remark}
	
	\subsection{Notation}
	\noindent The principal task of this section is to introduce some notations that will be involved in our future discussions in order to make calculation and proof convenient. These notations will appear in sequel for many times so we list them here for simplicity and brevity.
	\begin{itemize}
		\item $\beta_T(X) = \mathbf{E} \lbrack Y^T \vert X \rbrack$, $\beta_C(X) = \mathbf{E} \lbrack Y^C \vert X \rbrack$;
		
		\item $\beta(X) = \beta_T(X) - \beta_C(X)$;
		
		\item $\sigma_T^2(X) = \mathbf{var} (Y^T \vert X)$, $\sigma_C^2(X) = \mathbf{var} (Y^C \vert X)$;
		
		\item $\beta_T = \mathbf{E} \lbrack \beta_T(X) \rbrack$, $\beta_C = \mathbf{E} \lbrack \beta_C(X) \rbrack$, $p = \mathbf{E} \lbrack p(X) \rbrack$;
		
		\item $\beta = \mathbf{E} Y^T - \mathbf{E} Y^C = \mathbf{E} \lbrack \beta(X) \rbrack$.
	\end{itemize}
	
	\subsection{Covariates}
	\noindent Following the Neyman-Rubin model, the observed outcome $Y$ can be expressed as
	\begin{equation}
	Y = D Y^T + (1-D) Y^C,   \tag*{}
	\end{equation}
	hence covariates $X$ may influence outcome $Y$ through either treatment $D$ or potential outcomes $(Y^T, Y^C)$, or both. And it will be convenient for us to classify these covariates into the following 3 categories:
	\begin{enumerate}
		\item Treatment Predictors, or Exposure Predictors, denoted by $I$ in this paper, are also called Instrumental Variables. Instrumental variables are associated with exposures, but independent of potential outcomes, therefore they can only influence outcome $Y$ through treatment assignment $D$, for the instrumental variables, the following exclusion restriction condition holds:
		\begin{equation}
		(Y^T, Y^C) \perp I;     \tag*{}
		\end{equation}
		
		\item Confounding Factors, also known as Confounders, denoted by $U$ in this paper, are covariates that are associated with both exposure and outcome. Incorrect identification of confounders may violate the unconfoundedness assumption, and therefore cause severe bias of estimators;
		
		\item Predictors of Outcome, but not exposure, are usually denoted by $C$ in this paper. They are pre-treatment attributes of the individual that are unaffected by the treatment design. Its concrete probability characterization will be specified later.
	\end{enumerate}
	Now we denote $X=(I, U, C)$. A critical assumption for instrumental variables $I$ is that they are commonly supposed to be independent of other covariates, i.e., $I \perp (U, C)$. Under this assumption, the causal graph (Pearl, 2009)\cite{Pearl2009} of $D, I, U, C \ \text{and} \ Y$ can be illustrated as follow:
	\begin{equation}
	\xymatrix{
		& & U \\
		I & D & & Y & C
		\ar"2,1";"2,2",
		\ar"2,2";"2,4",
		\ar"2,5";"2,4",
		\ar"1,3";"2,2",
		\ar"1,3";"2,4"
	}                          \tag*{}
	\end{equation}
	\begin{center}
		Figure 1. The Directed Acyclic Graph for $(Y, D, I, U, C)$.
	\end{center}
	Therefore, the joint probability distribution of $(Y, D, I, U, C)$ can be specified as:
	\begin{equation}
	p(Y, D, I, U, C)=p(I)p(U, C)p(D \vert I, U)p(Y \vert D, U, C), \tag*{}
	\end{equation}
	based on the probability representation of Directed Acyclic Graph (Pearl, 2009)\cite{Pearl2009}. And we will first present the following lemma that will play an important role in demonstrating the variance reduction effect of including $C$ into the propensity score model:
	\begin{lemma}
		$C \perp D \vert I, U.$
	\end{lemma}
	\begin{proof}
		From the above joint probability decomposition, we know that
		\begin{equation}
		p(C, D, I, U)=p(I)p(U, C)p(D\vert I, U).    \tag*{}
		\end{equation}
		Since $I$ is independent of $(U, C)$, $p(I)p(U, C)=p(I, U, C)$, hence we have
		\begin{equation}
		p(C, D \vert I, U) = p(C \vert I, U) p(D \vert I, U),  \tag*{}
		\end{equation}
		therefore $C$ and $D$ are independent conditional on $(I, U)$.
	\end{proof}
	
	\section{Main Results}
	
	\subsection{Previous Results}
	\noindent A large number of non-parametric or semi-parametric estimators of average treatment effect have been proposed and examined thoroughly in existing literature, see Imbens (2004)\cite{Imbens2004} for a complete review. Most of them are $\sqrt{n}-\text{consistent}$ when the number of individuals in an observational study equals $n$, in the sense that the bias term $\mathbf{E} \hat{\beta} - \beta$ is of order $1/n$ in probability and the scaled error term $\sqrt{n}(\hat{\beta} - \mathbf{E} \hat{\beta})$ is asymptotically normal. Among such a variety of $\sqrt{n}-\text{consistent}$ estimators, some are asymptotically efficient while others are not. Here "asymptotically efficient" refers to reaching the semi-parametric efficiency bound. Newey (1990)\cite{Newey1990} reviewed the definition of "semi-parametric efficiency bound" and developed a systematical methodology to calculate it rigorously. After that, semi-parametric efficiency bound for estimators of average treatment effect in observational study has been obtained in Hahn (1998)\cite{Hahn1998}. Hahn (1998)\cite{Hahn1998} introduced the imputation estimator and proved its asymptotic efficiency when consistent estimators of propensity score function are used as inversed weights. Another asymptotically efficient estimator is the inverse probability weighting (IPW) estimator using estimated propensity scores proposed in Hirano et al. (2003)\cite{Hirano2003}. He suggested to use a Series Logit Estimator (SLE) for the propensity score instead of its true value, even when $p(X)$ itself is known, since using known propensity scores may do harm to its statistical efficiency.
	
	\subsubsection{The Imputation Estimator}
	\noindent The imputation estimator is defined as
	\begin{equation}
	\hat{\beta}_{imp}=\frac{1}{n} \sum_{i=1}^{n} (\frac{\hat{\mathbf{E}} \lbrack D_i Y_i \vert X_i \rbrack}{\hat{\mathbf{E}} \lbrack D_i \vert X_i \rbrack} - \frac{\hat{\mathbf{E}} \lbrack (1-D_i) Y_i \vert X_i \rbrack}{1 - \hat{\mathbf{E}} \lbrack D_i \vert X_i \rbrack}),   \tag*{}
	\end{equation}
	where $\hat{\mathbf{E}} \lbrack D_i Y_i \vert X_i \rbrack$, $\hat{\mathbf{E}} \lbrack (1-D_i) Y_i \vert X_i \rbrack$ and $\hat{\mathbf{E}} \lbrack D_i \vert X_i \rbrack$ denote nonparametric consistent estimators of their corresponding true conditional expectations. Since $\mathbf{E} \lbrack D_i Y_i \vert X_i \rbrack = \mathbf{E} \lbrack D_i \vert X_i \rbrack \mathbf{E} \lbrack Y_i^T \vert X_i \rbrack$ due to our unconfoundedness assumption, we may estimate $\beta_T(X_i) = \mathbf{E} \lbrack Y_i^T \vert X_i \rbrack$ by $\hat{\beta}_T(X_i) = \hat{\mathbf{E}} \lbrack D_i Y_i \vert X_i \rbrack / \hat{\mathbf{E}} \lbrack D_i \vert X_i \rbrack$, and estimate $\beta_C(X_i) = \mathbf{E} \lbrack Y_i^C \vert X_i \rbrack$ by $\hat{\beta}_C(X_i) = \hat{\mathbf{E}} \lbrack (1-D_i) Y_i \vert X_i \rbrack /(1 - \hat{\mathbf{E}} \lbrack D_i \vert X_i \rbrack)$. Note that $\hat{\beta}_T(X_i)$ and $\hat{\beta}_C(X_i)$ are $\sqrt{n}-\text{consistent}$ estimators of $\beta_T(X_i)$ and $\beta_C(X_i)$, respectively, and the fact that
	\begin{equation}
	\text{ATE}=\beta=\mathbf{E} \lbrack \beta_T(X) - \beta_C(X) \rbrack.   \tag*{}
	\end{equation}
	Therefore
	\begin{equation}
	\hat{\beta}_{imp} = \frac{1}{n} \sum_{i=1}^{n} (\hat{\beta}_T(X_i) - \hat{\beta}_C(X_i))      \tag*{}
	\end{equation}
	is $\sqrt{n}-\text{consistent}$ for estimating ATE.
	
	Hahn (1998)\cite{Hahn1998} showed that $\hat{\beta}_{imp}$ is efficient, i.e., asymptotic variance of $\hat{\beta}_{imp}$ is equal to the semi-parametric efficiency bound:
	\begin{equation}
	\mathbf{E} \lbrack \frac{\sigma_T^2(X)}{p(X)} + \frac{\sigma_C^2(X)}{1-p(X)} + (\beta(X)-\beta)^2 \rbrack.   \tag*{}
	\end{equation}      
	It is highlighted that even under the situation that $p(X)$ is known, it is necessary to use its semi-parametric estimators $\hat{p}(X_i)=\hat{\mathbf{E}} \lbrack D_i \vert X_i \rbrack$ instead of itself, since using the known value will do harm to its statistical efficiency. We will illustrate this fact under a rather simple setting, i.e., if the distribution of $X$, has a finite support. Under this circumstance, we may use the following empirical estimators:
	\begin{align*}
	& \hat{\mathbf{E}} \lbrack D_i Y_i \vert X_i=x \rbrack = \frac{\sum_{i=1}^{n} D_i Y_i \mathbf{I}_{\lbrace X_i=x \rbrace}}{\sum_{i=1}^{n} \mathbf{I}_{\lbrace X_i=x \rbrace}},  \\
	& \hat{\mathbf{E}} \lbrack (1-D_i) Y_i \vert X_i=x \rbrack = \frac{\sum_{i=1}^{n} (1-D_i) Y_i \mathbf{I}_{\lbrace X_i=x \rbrace}}{\sum_{i=1}^{n} \mathbf{I}_{\lbrace X_i=x \rbrace}},  \\
	& \hat{\mathbf{E}} \lbrack D_i \vert X_i=x \rbrack = \frac{\sum_{i=1}^{n} D_i \mathbf{I}_{\lbrace X_i=x \rbrace}}{\sum_{i=1}^{n} \mathbf{I}_{\lbrace X_i=x \rbrace}},
	\end{align*}
	and the known propensity score function to define another version of imputation estimator of ATE, i.e., 
	\begin{equation}
	\tilde{\beta}_{imp} = \frac{1}{n} \sum_{i=1}^{n} (\frac{\hat{\mathbf{E}} \lbrack D_i Y_i \vert X_i \rbrack}{p(X_i)} - \frac{\hat{\mathbf{E}} \lbrack (1-D_i) Y_i \vert X_i \rbrack}{1 - p(X_i)}).    \tag*{}
	\end{equation}
	Then we have the following theorem concerning on asymptotic efficiency of $\tilde{\beta}_{imp}$.
	\begin{theorem}
		Under Assumptions 1 and 2, and further assume that the propensity scores are known and the distribution of $X$ has a finite support, then $\sqrt{n} (\tilde{\beta}_{imp} - \beta)$ converges to a normal distribution in law, with mean $0$ and variance equal to
		\begin{equation}
		\mathbf{E} \lbrack \frac{\sigma_T^2(X)+\beta_T(X)^2}{p(X)} + \frac{\sigma_C^2(X)+\beta_C(X)^2}{1-p(X)} \rbrack - \beta^2,   \tag*{}
		\end{equation}
		which exceeds the semi-parametric efficiency bound by the magnitude of
		\begin{equation}
		\mathbf{E} \lbrack (\sqrt{\frac{1-p(X)}{p(X)}}\beta_T(X) + \sqrt{\frac{p(X)}{1-p(X)}}\beta_C(X))^2 \rbrack.   \tag*{}
		\end{equation}
		Hence, $\tilde{\beta}_{imp}$ is not efficient for estimating $\beta$.
	\end{theorem}
	(Proof in Appendix)
	
	\subsubsection{The IPW Estimator}
	\noindent The inverse probability weighting (IPW) estimator uses the reciprocal of propensity scores as weights to balance the distribution of potential outcomes between treatment group and control group, as its name suggests. It is formally defined as
	\begin{equation}
	\hat{\beta}_{ipw} = (1 / \sum_{i=1}^{n} \frac{D_i}{\hat{p}(X_i)}) \sum_{i=1}^{n} \frac{D_i Y_i}{\hat{p}(X_i)} - (1 / \sum_{i=1}^{n} \frac{1-D_i}{1-\hat{p}(X_i)}) \sum_{i=1}^{n} \frac{(1 - D_i) Y_i}{1 - \hat{p}(X_i)}.  \tag*{}
	\end{equation}
	Here we use a non-parametric estimator $\hat{p}(X)$ for $p(X)$, as discussed in Hirano et al. (2003)\cite{Hirano2003}. It can be shown that $\hat{\beta}_{ipw}$ is asymptotically efficient, with its asymptotic variance equal to the semi-parametric efficiency bound aforementioned.
	
	The situation is similar to the imputation estimator, and it is surprising that even when the propensity scores are known, using the true propensity scores the IPW estimator would not be fully efficient, and thus not very attractive in practical use (Hirano et al., 2003)\cite{Hirano2003}. For IPW estimator, knowing propensity score helps nothing in estimation and may be harmful sometimes. Such an estimator using known propensity scores, denoted by $\tilde{\beta}_{ipw}$, can be directly calculated from observational data after the weights are normalized to unity, that is:
	\begin{equation}
	\tilde{\beta}_{ipw} = (1 / \sum_{i=1}^{n} \frac{D_i}{p(X_i)}) \sum_{i=1}^{n} \frac{D_i Y_i}{p(X_i)} - (1 / \sum_{i=1}^{n} \frac{1-D_i}{1-p(X_i)}) \sum_{i=1}^{n} \frac{(1 - D_i) Y_i}{1 - p(X_i)}.  \tag*{}
	\end{equation}
	By usual argument we may figure out asymptotic distribution of $\tilde{\beta}_{ipw}$.
	\begin{theorem}
		Under Assumptions 1 and 2, and further assume that the propensity scores are known, then $\sqrt{n} (\tilde{\beta}_{ipw} - \beta)$ converges to a normal distribution in law, with mean $0$ and variance equal to
		\begin{equation}
		\mathbf{E} \lbrack \frac{\sigma_T^2(X)+(\beta_T(X)-\beta_T)^2}{p(X)} + \frac{\sigma_C^2(X)+(\beta_C(X)-\beta_C)^2}{1-p(X)} \rbrack,   \tag*{}
		\end{equation}
		which exceeds the semi-parametric efficiency bound by the magnitude of
		\begin{equation}
		\mathbf{E} \lbrack (\sqrt{\frac{1-p(X)}{p(X)}}(\beta_T(X)-\beta_T) + \sqrt{\frac{p(X)}{1-p(X)}}(\beta_C(X)-\beta_C))^2 \rbrack.   \tag*{}
		\end{equation}
		Hence, $\tilde{\beta}_{ipw}$ is not efficient for estimating $\beta$.
	\end{theorem}
	(Proof in Appendix)
	\begin{remark}
		Based on Theorem 1 and Theorem 2, we know that both $\tilde{\beta}_{imp}$ and $\tilde{\beta}_{ipw}$ are not efficient for estimating $\beta$. Therefore for these two estimators, knowing propensity score function does not help improve efficiency of estimators of $\beta$, and efficiently estimating $\beta$ requires using semi-parametric estimators of propensity scores under some regularity conditions.
	\end{remark}
	
	\subsection{Variance Reduction Effect}
	\noindent In this section, we consider the effect of including additional outcome predictors into the covariates. To be specific, we compare the efficiency of estimators of ATE when using $(I, U)$ as covariates with efficiency of estimators of ATE when using $X=(I, U, C)$ as covariates, here $C$ may be viewed as newly incorporated outcome predictors. First we show that the unconfoundedness assumption holds also for $(I, U)$, since there is no unmeasured confounder in $(I, U)$. We summarize this result in the following Lemma.
	\begin{lemma}
		$(Y^T, Y^C) \perp D \vert I, U$.
	\end{lemma}
	\begin{proof}
		Since $(Y^T, Y^C) \perp D \vert I, U, C$ and $C \perp D \vert I, U$, we know that
		\begin{align*}
		p(Y^T, Y^C, D, I, U, C) = & \frac{p(Y^T, Y^C, I, U, C)p(D, I, U, C)}{p(I, U, C)}  \\
		= & \frac{p(Y^T, Y^C, I, U, C)p(D, I, U)}{p(I, U)}.
		\end{align*}
		Take integration with respect to $C$ on both sides, we find that
		\begin{equation}
		p(Y^T, Y^C, D, I, U) = \frac{p(Y^T, Y^C, I, U)p(D, I, U)}{p(I, U)},     \tag*{}
		\end{equation}
		therefore $(Y^T, Y^C) \perp D \vert I, U$. End of proof.
	\end{proof}
	Now since $p(I, U) = p(X)$, the Overlap assumption also holds for $(I, U)$, to facilitate further discussions we denote $X_0=(I, U)$. Then we investigate the variance reduction effect of including $C$ into $X$, compared with the situation in which we merely use $X_0$ in propensity score matching. To begin with, we first introduce two useful lemmas.
	\begin{lemma}
		Let $f$ and $g$ be integrable random variables, then we have
		\begin{equation}
		\mathbf{E} \lbrack \frac{\mathbf{E} \lbrack f \vert X_0 \rbrack}{p(X_0)} + \frac{\mathbf{E} \lbrack g \vert X_0 \rbrack}{1-p(X_0)} \rbrack =
		\mathbf{E} \lbrack \frac{\mathbf{E} \lbrack f \vert X \rbrack}{p(X)} + \frac{\mathbf{E} \lbrack g \vert X \rbrack}{1-p(X)} \rbrack.   \tag*{}
		\end{equation}
	\end{lemma}
	(Proof in Appendix)
	\begin{lemma}
		Let $f$ and $g$ be integrable random variables, then we have
		\begin{equation}
		\mathbf{E} \lbrack (\sqrt{\frac{1-p(X_0)}{p(X_0)}} \mathbf{E} \lbrack f \vert X_0 \rbrack + \sqrt{\frac{p(X_0)}{1-p(X_0)}} \mathbf{E} \lbrack g \vert X_0 \rbrack)^2 \rbrack \leq
		\mathbf{E} \lbrack (\sqrt{\frac{1-p(X)}{p(X)}} \mathbf{E} \lbrack f \vert X \rbrack + \sqrt{\frac{p(X)}{1-p(X)}} \mathbf{E} \lbrack g \vert X \rbrack)^2 \rbrack.   \tag*{}
		\end{equation}
	\end{lemma}
	(Proof in Appendix)
	
	Now we first show that, for the imputation estimator and IPW estimator using known propensity scores that are not efficient, i.e., $\tilde{\beta}_{imp}$ and $\tilde{\beta}_{ipw}$, including $C$ into $X$ neither increase or decrease their asymptotic variances. See the following theorem.
	\begin{theorem}
		Using the same notations as before, then the asymptotic variances of $\tilde{\beta}_{imp}$ and $\tilde{\beta}_{ipw}$ both keep invariant after including $C$ in $X$, i.e.,
		\begin{equation}
		\mathbf{E} \lbrack \frac{\sigma_T^2(X_0)+\beta_T(X_0)^2}{p(X_0)} + \frac{\sigma_C^2(X_0)+\beta_C(X_0)^2}{1-p(X_0)} \rbrack = \mathbf{E} \lbrack \frac{\sigma_T^2(X)+\beta_T(X)^2}{p(X)} + \frac{\sigma_C^2(X)+\beta_C(X)^2}{1-p(X)} \rbrack, \tag*{}
		\end{equation}
		and
		\begin{align*}
		& \mathbf{E} \lbrack \frac{\sigma_T^2(X_0)+(\beta_T(X_0)-\beta_T)^2}{p(X_0)} + \frac{\sigma_C^2(X_0)+(\beta_C(X_0)-\beta_C)^2}{1-p(X_0)} \rbrack   \\
		= & \mathbf{E} \lbrack \frac{\sigma_T^2(X)+(\beta_T(X)-\beta_T)^2}{p(X)} + \frac{\sigma_C^2(X)+(\beta_C(X)-\beta_C)^2}{1-p(X)} \rbrack.
		\end{align*}
	\end{theorem}
	\begin{proof}
		We only need to note that since $\beta_T(X) = \mathbf{E} \lbrack Y^T \vert X \rbrack$, and $\sigma_T^2(X) + \beta_T(X)^2 = \mathbf{E} \lbrack (Y^T)^2 \vert X \rbrack$, and the similar holds for the control group, then we have
		\begin{equation}
		\mathbf{E} \lbrack \frac{\sigma_T^2(X)+\beta_T(X)^2}{p(X)} + \frac{\sigma_C^2(X)+\beta_C(X)^2}{1-p(X)} \rbrack = \mathbf{E} \lbrack \frac{\mathbf{E} \lbrack (Y^T)^2 \vert X \rbrack}{p(X)} + \frac{\mathbf{E} \lbrack (Y^C)^2 \vert X \rbrack}{1-p(X)} \rbrack,      \tag*{}
		\end{equation}
		then using Lemma 3 yields the first identity. Now still with the aid of Lemma 3, we show the second equality by directly calculation:
		\begin{equation}
		\mathbf{E} \lbrack \frac{\sigma_T^2(X)+(\beta_T(X)-\beta_T)^2}{p(X)} + \frac{\sigma_C^2(X)+(\beta_C(X)-\beta_C)^2}{1-p(X)} \rbrack = \mathbf{E} \lbrack \frac{\mathbf{E} \lbrack (Y^T-\beta_T)^2 \vert X \rbrack}{p(X)} + \frac{\mathbf{E} \lbrack (Y^C-\beta_C)^2 \vert X \rbrack}{1-p(X)} \rbrack.  \tag*{}
		\end{equation}
		End of proof.
	\end{proof}
	In the meantime, we claim that, incorporating $C$ into $X$ can always help lower the semi-parametric efficiency bound, hence improve statistical efficiency of the efficient estimators, $\hat{\beta}_{imp}$ and $\hat{\beta}_{ipw}$.
	\begin{theorem}
		Using the same notations as before, then the asymptotic efficiency bound becomes smaller after including $C$ in $X$, i.e.,
		\begin{equation}
		\mathbf{E} \lbrack \frac{\sigma_T^2(X_0)}{p(X_0)} + \frac{\sigma_C^2(X_0)}{1-p(X_0)} + (\beta(X_0)-\beta)^2 \rbrack \geq \mathbf{E} \lbrack \frac{\sigma_T^2(X)}{p(X)} + \frac{\sigma_C^2(X)}{1-p(X)} + (\beta(X)-\beta)^2 \rbrack.   \tag*{}
		\end{equation}
	\end{theorem}
	\begin{proof}
		According to Theorem 1, if we choose $X_0$ as the covariates to be included into propensity score matching, the difference between asymptotic variance of $\tilde{\beta}_{imp}$ and the asymptotic efficiency bound is
		\begin{equation}
		\mathbf{E} \lbrack (\sqrt{\frac{1-p(X_0)}{p(X_0)}}\beta_T(X_0) + \sqrt{\frac{p(X_0)}{1-p(X_0)}}\beta_C(X_0))^2 \rbrack, \tag*{}
		\end{equation}
		which will increase after we additionally incorporate $C$ into $X$ due to the conclusion of Lemma 4, since asymptotic variance of $\tilde{\beta}_{imp}$ will keep invariant based on Theorem 3, we know that the asymptotic efficiency bound will decrease, end of proof.
	\end{proof}
	\begin{remark}
		It is worth noting that the variance reduction effect of incorporating outcome predictors into the covariates only appears in obtaining asymptotically efficient estimators. For $\tilde{\beta}_{imp}$ and $\tilde{\beta}_{ipw}$, outcome predictors have no impact on their statistical efficiency. After all, at least they are not harmful to efficient estimation of ATE under unconfoundedness assumption.
	\end{remark}
	
	\subsection{Variance Inflation Effect}
	\noindent In this section, we deal with instrumental variables that are independent of potential outcomes conditional on the confounders and outcome predictors, i.e.,
	\begin{equation}
	(Y^T, Y^C) \perp I \vert U, C,   \tag*{}
	\end{equation}
	and we will show that including such kind of instrumental variables into the propensity score function will cause variance inflation for the four imputation estimators and IPW estimators aforementioned. Here, $I$ may be refer to the instrumental variables newly incorporated into $(U, C)$, and we aim to illustrate the variance inflation effect caused by additionally including $I$. And we first show that, the unconfoundedness assumption holds for $(U, C)$ as well since $(U, C)$ contains no unmeasured confounders.
	\begin{lemma}
		If we assume that $(Y^T, Y^C) \perp I \vert U, C$, then $(Y^T, Y^C) \perp D \vert U, C$.
	\end{lemma}
	\begin{proof}
		From the unconfoundedness assumption and $(Y^T, Y^C) \perp I \vert U, C$, we know that
		\begin{align*}
		p(Y^T, Y^C, D, I, U, C) = & \frac{p(Y^T, Y^C, I, U, C) p(D, I, U, C)}{p(I, U, C)}   \\
		= & \frac{p(Y^T, Y^C, U, C) p(D, I, U, C)}{p(U, C)}.
		\end{align*}
		Now integrate with respect to $I$ on both sides, we obtain that
		\begin{equation}
		p(Y^T, Y^C, D, U, C) = \frac{p(Y^T, Y^C, U, C) p(D, U, C)}{p(U, C)},   \tag*{}
		\end{equation}
		which is equivalent to the fact that
		\begin{equation}
		(Y^T, Y^C) \perp D \vert U, C.  \tag*{}
		\end{equation}
		End of proof.
	\end{proof}
	Now, to simplify the notation we denote $X_1=(U, C)$, since
	\begin{equation}
	p(X_1) = \mathbf{E} \lbrack D \vert X_1 \rbrack = \mathbf{E} \lbrack \mathbf{E} \lbrack D \vert X \rbrack \vert X_1 \rbrack = \mathbf{E} \lbrack p(X) \vert X_1 \rbrack,         \tag*{}
	\end{equation}
	The Overlap assumption (Assumption 2) also holds for $X_1$. Therefore, if we use $X_1$ instead of $X$ as covariates to be adjusted for in estimation of ATE, the semi-parametric efficiency bound will become
	\begin{equation}
	\mathbf{E} \lbrack \frac{\sigma_T^2(X_1)}{p(X_1)} + \frac{\sigma_C^2(X_1)}{1-p(X_1)} + (\beta(X_1)-\beta)^2 \rbrack,  \tag*{}
	\end{equation}
	Based on Part 3 of Hahn (2004)\cite{Hahn2004}, we can prove the following theorem comparing the above two efficiency bounds, and then illustrate the variance inflation effect on the asymptotically efficient estimators $\hat{\beta}_{imp}$ and $\hat{\beta}_{ipw}$ after including $I$ in $X$.
	\begin{theorem}
		Using the same notations as before, then the asymptotic efficiency bound becomes larger after including $I$ in $X$, i.e.,
		\begin{equation}
		\mathbf{E} \lbrack \frac{\sigma_T^2(X_1)}{p(X_1)} + \frac{\sigma_C^2(X_1)}{1-p(X_1)} + (\beta(X_1)-\beta)^2 \rbrack \leq \mathbf{E} \lbrack \frac{\sigma_T^2(X)}{p(X)} + \frac{\sigma_C^2(X)}{1-p(X)} + (\beta(X)-\beta)^2 \rbrack.   \tag*{}
		\end{equation}
	\end{theorem}
	\begin{proof}
		According to Lemma 1, $p(Y^T, Y^C \vert X)=p(Y^T, Y^C \vert X_1)$. Combined with Hahn (2004, Part 3)\cite{Hahn2004} the Theorem can be proved.
	\end{proof}
	In fact, the similar results also hold for $\tilde{\beta}_{imp}$ and $\tilde{\beta}_{ipw}$. To be specific, consider the influence of including $I$ in $X$ on statistical efficiency of $\tilde{\beta}_{imp}$ and $\tilde{\beta}_{ipw}$, we also have
	\begin{theorem}
		Using the same notations as before, then the asymptotic variances of $\tilde{\beta}_{imp}$ and $\tilde{\beta}_{ipw}$ both become larger after including $I$ in $X$, i.e.,
		\begin{equation}
		\mathbf{E} \lbrack \frac{\sigma_T^2(X_1)+\beta_T(X_1)^2}{p(X_1)} + \frac{\sigma_C^2(X_1)+\beta_C(X_1)^2}{1-p(X_1)} \rbrack \leq \mathbf{E} \lbrack \frac{\sigma_T^2(X)+\beta_T(X)^2}{p(X)} + \frac{\sigma_C^2(X)+\beta_C(X)^2}{1-p(X)} \rbrack, \tag*{}
		\end{equation}
		and
		\begin{align*}
		& \mathbf{E} \lbrack      \frac{\sigma_T^2(X_1)+(\beta_T(X_1)-\beta_T)^2}{p(X_1)} + \frac{\sigma_C^2(X_1)+(\beta_C(X_1)-\beta_C)^2}{1-p(X_1)} \rbrack      \\
		\leq & \mathbf{E} \lbrack \frac{\sigma_T^2(X)+(\beta_T(X)-\beta_T)^2}{p(X)} + \frac{\sigma_C^2(X)+(\beta_C(X)-\beta_C)^2}{1-p(X)} \rbrack.
		\end{align*}
	\end{theorem}
	(Proof in Appendix)
	
	\section{The Linearly Modified Estimator and Its Properties}
	\noindent In this section, we will first review the KPS estimator proposed in Rothe (2016), which aims to solve the "curse of dimensionality" issue encountered in estimating propensity scores which is necessary for calculating $\hat{\beta}_{imp}$ and $\hat{\beta}_{ipw}$. Then we may introduce our linearly modified (LM) estimator and discuss its properties. Since both KPS estimator and LM estimator require known propensity score function, throughout this section we make the following assumption:
	\begin{assumption}
		The propensity score function, $p(X)$, is known.
	\end{assumption}
	
	\subsection{The KPS Estimator}
	\noindent The value of knowing the propensity score in estimating ATE is highlighted in Rothe (2016)\cite{Rothe2016}, although many existing results seem to suggest that it is not necessary and sometimes even harmful to try to obtain knowledge of propensity score for solely purpose of estimation. Rothe (2016)\cite{Rothe2016} proposed a "known propensity score" (KPS) estimator, denoted by $\hat{\beta}_{kps}$, and was built upon the ideas from literature investigating double robustness of estimators (e.g. Robins et al., 1994\cite{Robins1994}; Robins et al., 1995\cite{Robins1995}). The KPS estimator takes the form of
	\begin{equation}
	\hat{\beta}_{kps} = \frac{1}{n} \sum_{i=1}^{n} (\frac{D_i Y_i}{p(X_i)} - \frac{(1-D_i) Y_i}{1-p(X_i)} - (D_i - p(X_i)) (\frac{\hat{\beta}_T(X_i)}{p(X_i)} - \frac{\hat{\beta}_C(X_i)}{1-p(X_i)})),         \tag*{}
	\end{equation}
	where $\hat{\beta}_T(X_i)$ and $\hat{\beta}_C(X_i)$ are non-parametric estimators of $\beta_T(X_i)$ and $\beta_C(X_i)$ respectively. Rothe (2016)\cite{Rothe2016} argued that as long as $\beta_T(X_i)$ and $\beta_C(X_i)$ can be consistently estimated, the KPS estimator is fully efficient. Furthermore, compared to previously introduced efficient estimators that require consistently estimating the propensity scores, the KPS estimator ask a rather low degree of accuracy on estimation and rather mild regularity conditions (Rothe, 2016)\cite{Rothe2016}.
	
	\subsection{The Rationale Behind our Linearly Modified Estimator}
	\noindent The rationale behind our linearly modified (LM) estimator, is that on the one hand all aforementioned estimators that are completely efficient, including the KPS estimator, require either $p(X)$ or $(\beta_T(X), \beta_C(X))$ to be estimated consistently, no matter to what extent certain regularity conditions are supposed to be satisfied. This could be difficult and computationally expensive, especially under the general setting that we do not make any specific assumptions on underlying probability model. On the other hand, $\tilde{\beta}_{ipw}$ only depends on the known propensity score and observed data and thus is easy to calculate, but not fully efficient and sometimes may have poor performance compared to efficient estimators of ATE.
	
	The LM estimator applies a linear modification to $\tilde{\beta}_{ipw}$ using the difference of weighted average of covariates between treatment group and control group, and it can solve both issues addressed in a large degree. It is similar to the interaction estimator in Lin (2013)\cite{Lin2013}, that utilizes covariates as ancillary to improve asymptotic precision of the intention-to-treat (ITT) estimator (Freedman, 2008)\cite{Freedman2008}. To be specific, define
	\begin{equation}
	\tilde{x}_{ipw} = (1 / \sum_{i=1}^{n} \frac{D_i}{p(X_i)}) \sum_{i=1}^{n} \frac{D_i X_i}{p(X_i)} - (1 / \sum_{i=1}^{n} \frac{1-D_i}{1-p(X_i)}) \sum_{i=1}^{n} \frac{(1 - D_i) X_i}{1 - p(X_i)}.  \tag*{}
	\end{equation}
	And we may expect that the information of $X$ provided by $\tilde{x}_{ipw}$ is helpful in estimating ATE, especially help to reduce the proportion of asymptotic variance of $\tilde{x}_{ipw}$ that exceeds the semi-parametric efficiency bound. If $X \in \mathbb{R}^K$, then $\tilde{x}_{ipw} \in \mathbb{R}^K$. In following discussion, we use $\text{asycov} (\mathbf{x}, \mathbf{y})$ to denote asymptotic covariance matrix between random vectors $\mathbf{x}$ and $\mathbf{y}$, and $\text{asyvar} (x)$ to denote asymptotic variance of random variable $x$.
	\begin{lemma}
		Under Assumptions 1, 2 and 3, we have
		\begin{align*}
		& \text{asycov} (\tilde{x}_{ipw}, \tilde{x}_{ipw}) = \mathbf{E} \lbrack \frac{1}{p(X)(1-p(X))} (X-\mathbf{E}X) (X-\mathbf{E}X)^T \rbrack, \\
		& \text{asycov} (\tilde{x}_{ipw}, \tilde{\beta}_{ipw}) = \mathbf{E} \lbrack (X-\mathbf{E}X) (\frac{1}{p(X)} (\beta_T(X)-\beta_T) + \frac{1}{1-p(X)} (\beta_C(X)-\beta_C)) \rbrack.
		\end{align*}
	\end{lemma}
	(Proof in Appendix)
	
	Based on Lemma 3, we can show the following powerful proposition that helps us to construct $\sqrt{n}-\text{consistent}$ estimators from $\tilde{\beta}_{ipw}$:
	\begin{prop}
		Let $\alpha \in \mathbb{R}^K$ be any constant vector, if $\hat{\alpha}$ is a $\sqrt{n}-\text{consistent}$ estimator of $\alpha$, i.e., $\hat{\alpha}-\alpha = O_P(1/\sqrt{n})$, then as $n$ grows to infinity, $\sqrt{n} (\tilde{\beta}_{ipw} - \hat{\alpha}^T \tilde{x}_{ipw} - \beta)$ converges to a normal distribution in law, with mean $0$ and variance equal to
		\begin{equation}
		\text{asyvar}(\tilde{\beta}_{ipw}) - 2 \alpha^T \text{asycov} (\tilde{x}_{ipw}, \tilde{\beta}_{ipw}) + \alpha^T \text{asycov} (\tilde{x}_{ipw}, \tilde{x}_{ipw}) \alpha.  \tag*{}
		\end{equation}
	\end{prop}
	(Proof in Appendix)
	
	Hence, in order to minimize the asymptotic variance of such $\sqrt{n}-\text{consistent}$ estimator of ATE in the form of $\tilde{\beta}_{ipw} - \hat{\alpha}^T \tilde{x}_{ipw}$, i.e., a linear modification of $\tilde{\beta}_{ipw}$ using $\tilde{x}_{ipw}$, it's nature to take
	\begin{equation}
	\alpha = \text{asycov} (\tilde{x}_{ipw}, \tilde{x}_{ipw})^{-1} \text{asycov} (\tilde{x}_{ipw}, \tilde{\beta}_{ipw}),    \tag*{}
	\end{equation}
	and find an estimator $\hat{\alpha}$ of $\alpha$ such that $\hat{\alpha}-\alpha = O_P(1/\sqrt{n})$. It suffices to find empirical estimators of $\text{asycov} (\tilde{x}_{ipw}, \tilde{x}_{ipw})$ and $\text{asycov} (\tilde{x}_{ipw}, \tilde{\beta}_{ipw})$ that are $\sqrt{n}-\text{consistent}$ respectively. The following Proposition 2 gives our desired estimators.
	\begin{prop}
		Let $\bar{X} = (1/n) \sum_{i=1}^{n} X_i$ denote sample mean of covariates, then
		\begin{equation}
		\widehat{\text{asycov}} (\tilde{x}_{ipw}, \tilde{x}_{ipw}) = \frac{1}{n} \sum_{i=1}^{n} \frac{1}{p(X_i) (1-p(X_i))} (X_i - \bar{X}) (X_i - \bar{X})^T   \tag*{}
		\end{equation}
		is a $\sqrt{n}-\text{consistent}$ estimator of $\text{asycov} (\tilde{x}_{ipw}, \tilde{x}_{ipw})$, in the same time,
		\begin{equation}
		\widehat{\text{asycov}} (\tilde{x}_{ipw}, \tilde{\beta}_{ipw}) = \frac{1}{n} \sum_{j=1}^{n} (X_j - \bar{X}) (\frac{D_j}{p(X_j)^2} (Y_j^T - \frac{1}{n} \sum_{i=1}^{n} \frac{D_i Y_i}{p(X_i)}) + \frac{1-D_j}{(1-p(X_j))^2} (Y_j^C - \frac{1}{n} \sum_{i=1}^{n} \frac{(1-D_i) Y_i}{1-p(X_i)}))  \tag*{}
		\end{equation}
		is a $\sqrt{n}-\text{consistent}$ estimator of $\text{asycov} (\tilde{x}_{ipw}, \tilde{\beta}_{ipw})$.
	\end{prop}
	(Proof in Appendix)
	
	Now it's time for us to define our LM estimator, choose
	\begin{equation}
	\hat{\alpha} = \widehat{\text{asycov}} (\tilde{x}_{ipw}, \tilde{x}_{ipw})^{-1} \widehat{\text{asycov}} (\tilde{x}_{ipw}, \tilde{\beta}_{ipw}),   \tag*{}
	\end{equation}
	and then define $\tilde{\beta}_{lm} = \tilde{\beta}_{ipw} - \hat{\alpha}^T \tilde{x}_{ipw}$ to be the LM estimator. Its statistical properties will be analyzed in next section.
	
	\subsection{Statistical Properties}
	\noindent In this section, we investigate some of $\tilde{\beta}_{lm}$'s important properties, as is well known that since the linear correction term $\tilde{x}_{ipw}$ is of order $1/\sqrt{n}$ in probability, $\tilde{\beta}_{lm}$ is still $\sqrt{n}-\text{consistent}$ and asymptotically normal, and its asymptotic variance can be calculated directly from Proposition 1. Concerning its asymptotic efficiency, we have the following theorem.
	\begin{theorem}
		Under Assumptions 1, 2 and 3, $\tilde{\beta}_{lm}$ is more efficient than $\tilde{\beta}_{ipw}$ in estimating ATE, and strictly more efficient unless
		\begin{equation}
		\text{asycov} (\tilde{x}_{ipw}, \tilde{\beta}_{ipw}) = 0,    \tag*{}
		\end{equation}
		which is equivalent to the fact that
		\begin{equation}
		\frac{1}{p(X)} (\beta_T(X)-\beta_T) + \frac{1}{1-p(X)} (\beta_C(X)-\beta_C) \ \text{is uncorrelated with} \ X.       \tag*{}
		\end{equation}
	\end{theorem}
	(Proof in Appendix)
	
	Since $\tilde{\beta}_{lm}$ only relies on the known propensity score function and observed data, we can not expect that its asymptotic variance will reach the semi-parametric efficiency bound. However, this is nearly true when $\beta_T(X)$ and $\beta_C(X)$ are both strongly correlated with $X$, which is very common in most circumstances. We will use some numerical examples to illustrate this point in next section. In fact, $\tilde{\beta}_{lm}$ only suffers from small efficiency loss compared to efficient estimators of ATE while enjoys a great convenience compared to them since we need not calculate non-parametric estimators of propensity score function and conditional expectations of potential outcomes. Sometimes $\tilde{\beta}_{lm}$ is fully efficient, see the following theorem.
	\begin{theorem}
		If there exist $a, b \in \mathbb{R}$ and $c \in \mathbb{R}^K$ such that for the potential outcomes $Y^T$ and $Y^C$,
		\begin{equation}
		\beta_T(X)=\mathbf{E} \lbrack Y^T \vert X \rbrack = a + b + c^T X, \beta_C(X)=\mathbf{E} \lbrack Y^C \vert X \rbrack = a + c^T X.  \tag*{}
		\end{equation}
		Then $\tilde{\beta}_{lm}$ is asymptotically efficient.
	\end{theorem}
	(Proof in Appendix)
	
	\section{Simulation Studies}
	\noindent In this section, we present some simulation studies to illustrate the variance reduction effect of linearly modified estimator $\tilde{\beta}_{lm}$ compared to $\tilde{\beta}_{ipw}$. We will show that in most general settings, $\tilde{\beta}_{lm}$ reduces a considerable proportion of the difference between $\text{asyvar}(\tilde{\beta}_{ipw})$ and semi-parametric efficiency bound, i.e., although $\tilde{\beta}_{lm}$ is still not asymptotically efficient, the difference between $\text{asyvar}(\tilde{\beta}_{lm})$ and semi-parametric bound is much smaller than the difference between $\text{asyvar}(\tilde{\beta}_{ipw})$ and semi-parametric efficiency bound. This fact will strongly support our motivation behind construction of $\tilde{\beta}_{lm}$ and further provide us with a decent method to obtain a nearly efficient estimator of ATE, especially when estimating propensity score function $p(X)$ is difficult or the corresponding regularity conditions are hard be satisfied, when the propensity score function is supposed to be known.
	
	\subsection{Basic Settings}
	\noindent We investigate efficiency gain of $\tilde{\beta}_{lm}$ from $\tilde{\beta}_{ipw}$ under a quite general circumstance, we assume that the potential outcomes satisfy a linear model and the propensity score function is in the logistic form. To be specific, for potential outcomes $Y^T$ and $Y^C$, we assume that
	\begin{equation}
	\beta_T(X) = a_1 + c_1^T X, \ \beta_C(X) = a_0 + c_0^T X, \tag*{}
	\end{equation}
	as for propensity score function, we assume that
	\begin{equation}
	p(X) = \frac{\exp(\gamma^T X + 1)}{1 + \exp(\gamma^T X + 1)}.  \tag*{}
	\end{equation}
	Here, we assume that $X=(I, U, C)$, where $I$, $U$ and $C$ are independent and identically distributed random variables. For both simplicity and integrity of our discussion, we will perform simulation studies under situations when their common distribution is uniform or normal with zero mean respectively.
	
	We mainly focus on two quantities that explain this variance reduction effect of $\tilde{\beta}_{lm}$ entirely in numerical simulation, the first one is the difference between $\text{asyvar}(\tilde{\beta}_{ipw})$ and semi-parametric efficiency bound, while the second is the difference between $\text{asyvar}(\tilde{\beta}_{ipw})$ and $\text{asyvar}(\tilde{\beta}_{lm})$. The latter one is always smaller than the former one since $\tilde{\beta}_{lm}$ is not always fully efficient. According to what we've shown in Theorem 8, if the outcome $Y$ satisfies a linear model in $D$ and $X$, or equivalently, $c_0=c_1$, then $\tilde{\beta}_{lm}$ is asymptotically efficient. In our simulation, we fix $c_1=(0, 0, 1)^T$ and let $c_0$ varies. Specifically speaking, for $\theta \in \lbrack 0, 2\pi )$, set
	\begin{equation}
	c_0(\theta) = (0, \sin{\theta}, \cos{\theta})^T,   \tag*{}
	\end{equation}
	respectively, recall that $X=(I, U, C)$, then we have
	\begin{equation}
	\beta_T(X) - \beta_T = C, \ \beta_C(X) - \beta_C = U \sin{\theta} + C \cos{\theta}    \tag*{}
	\end{equation}
	do not depend on instrumental variable $I$, hence the exclusion restriction is valid. We then choose $\gamma = (t, t, 0)$ for different values of $t$, i.e.,
	\begin{equation}
	p(X) = \frac{\exp(t(I+U) + 1)}{1 + \exp(t(I+U) + 1)}  \tag*{}
	\end{equation}
	do not depend on outcome predictor $C$. According to the Overlap assumption, the magnitude of $t$ can not be set too large since this will disobey the randomness of treatment assignment. When the common distribution of $I$, $U$ and $C$ is $\mathcal{U} \lbrack -1, 1 \rbrack$, we take $t=2, 1, 0.5$ respectively, when the common distribution of $I$, $U$ and $C$ is $\mathcal{N} (0, 1)$, we take $t=1, 0.5, 0.25$ respectively. 
	
	Now given the underlying distribution of covariates $X$, for each choice of pair of $t$ and $\theta \in \lbrack 0, 2\pi )$, we utilize Monte Carlo method to compute the difference between $\text{asyvar}(\tilde{\beta}_{ipw})$ and $\text{asyvar}(\tilde{\beta}_{lm})$, and the difference between $\text{asyvar}(\tilde{\beta}_{ipw})$ and semi-parametric efficiency bound, respectively. Then we calculate their ratio and denote that ratio by $R(\theta, t)$, then we have
	\begin{equation}
	R(\theta, t) = \frac{\text{asyvar}(\tilde{\beta}_{ipw}) - \text{asyvar}(\tilde{\beta}_{lm})}{\text{asyvar}(\tilde{\beta}_{ipw}) - \text{semi-parametric efficiency bound}}.  \tag*{}
	\end{equation}
	We then plot $R(\theta, t)$ versus $\theta$, for different choices of $t$. Then we calculate the following average reduced proportion
	\begin{equation}
	R(t) = \frac{1}{2 \pi} \int_{0}^{2 \pi} R(\theta, t) \mathrm{d} \theta \tag*{}
	\end{equation}
	for different values of $t$. Numerical results will be presented in next subsection.
	
	\subsection{Numerical Results}
	\noindent In this subsection, we present numerical results concerning plot of $R(\theta, t)$ versus $\theta$ and value of $R(t)$ for each covariates' distribution and choice of $t$. For $R(\theta, t)-\theta$ plot, see the following figures. The title "Uniform distribution" indicates that the common distribution of $I$, $U$ and $C$ is $\mathcal{U} \lbrack -1, 1 \rbrack$, while the title "Normal distribution" indicates that the common distribution of $I$, $U$ and $C$ is $\mathcal{N} (0, 1)$.
	\begin{center}
		\includegraphics[width=27em]{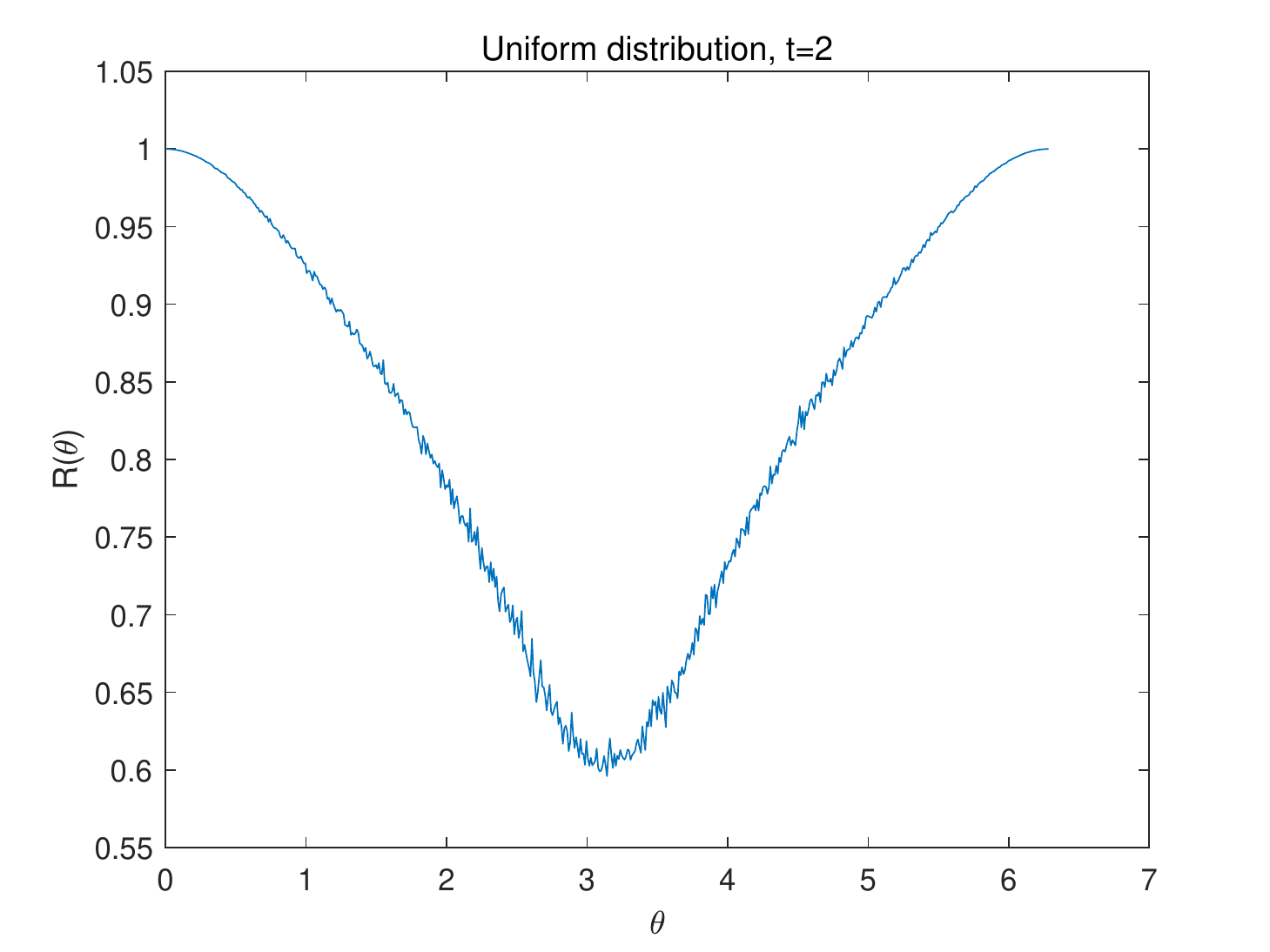}
	\end{center}
	\begin{center}
		Figure 2. $R(\theta, t)-\theta$ plot under covariates' distribution $\mathcal{U} \lbrack -1, 1 \rbrack$ and $t=2$.
	\end{center}
	\begin{center}
		\includegraphics[width=27em]{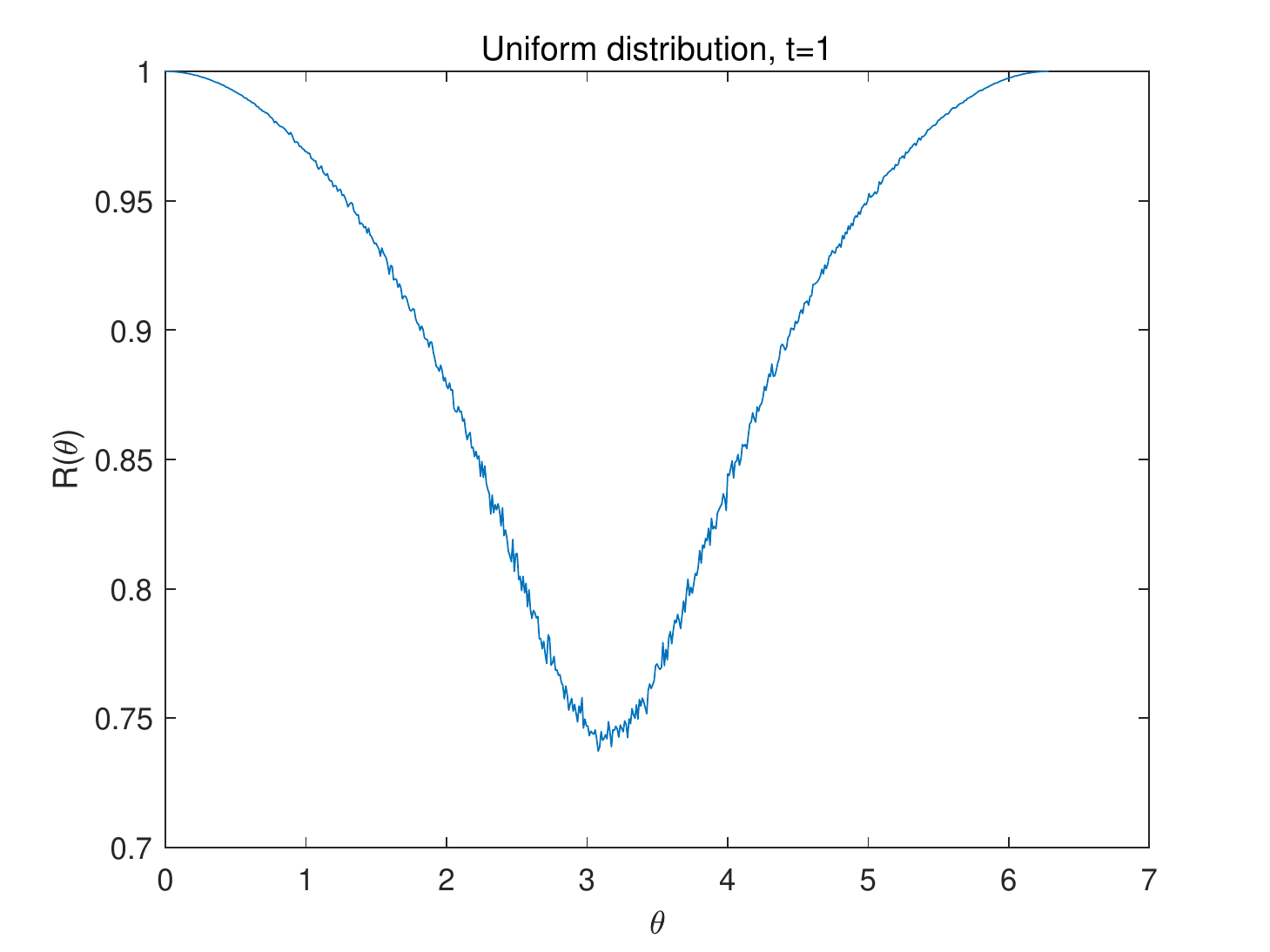}
	\end{center}
	\begin{center}
		Figure 3. $R(\theta, t)-\theta$ plot under covariates' distribution $\mathcal{U} \lbrack -1, 1 \rbrack$ and $t=1$.
	\end{center}
	\begin{center}
		\includegraphics[width=27em]{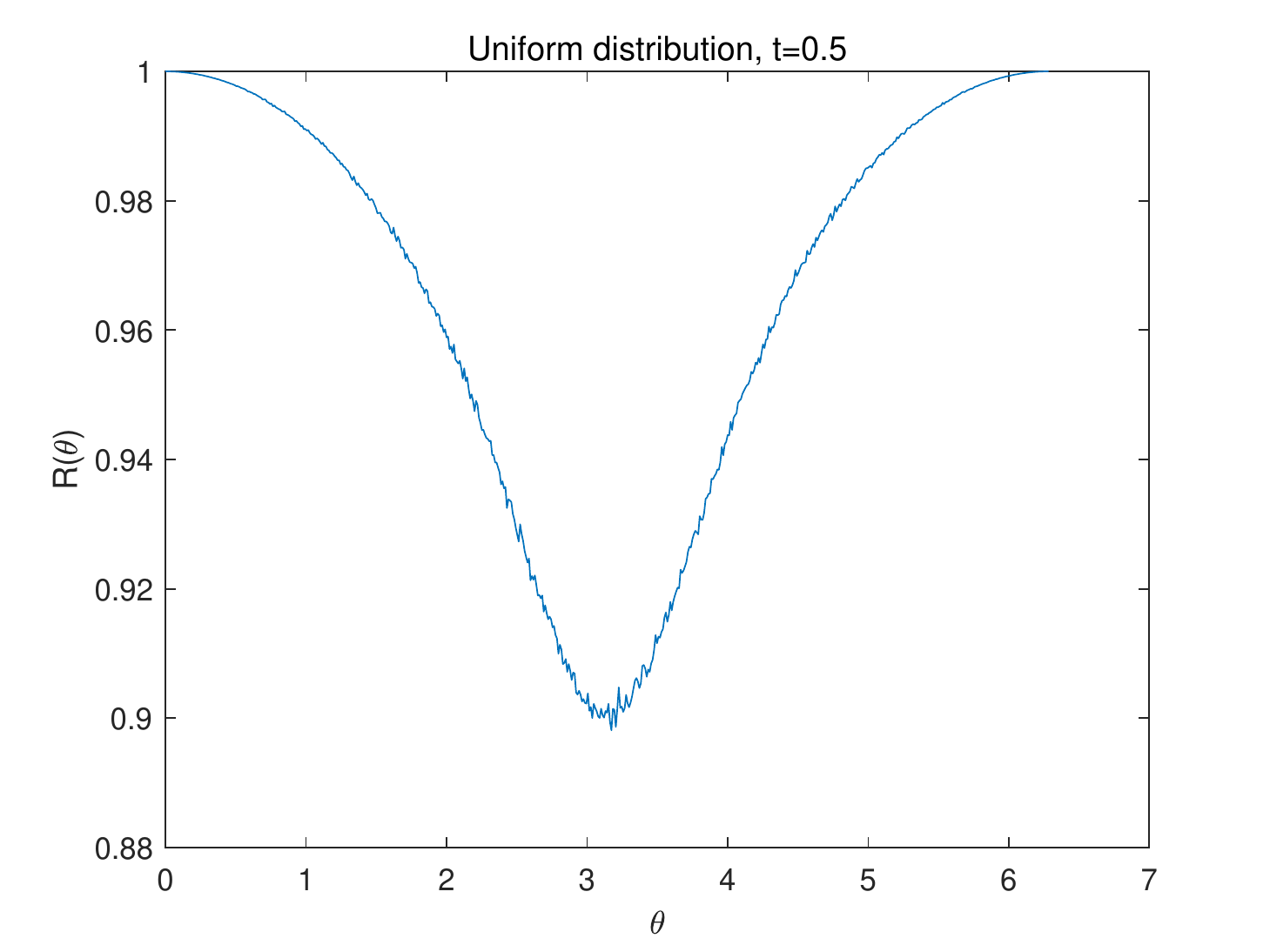}
	\end{center}
	\begin{center}
		Figure 4. $R(\theta, t)-\theta$ plot under covariates' distribution $\mathcal{U} \lbrack -1, 1 \rbrack$ and $t=0.5$.
	\end{center}
	\begin{center}
		\includegraphics[width=27em]{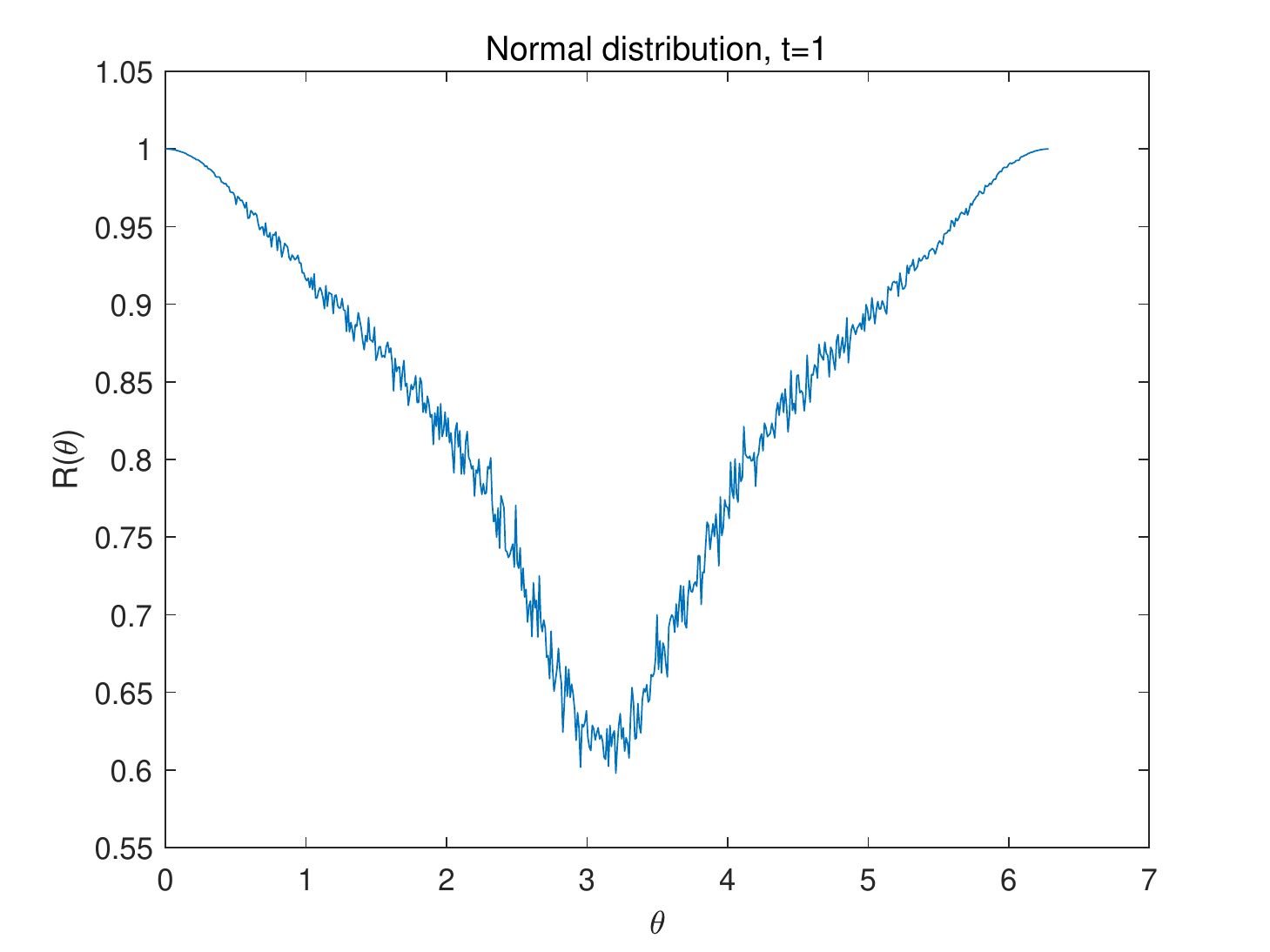}
	\end{center}
	\begin{center}
		Figure 5. $R(\theta, t)-\theta$ plot under covariates' distribution $\mathcal{N} (0, 1)$ and $t=1$.
	\end{center}
	\begin{center}
		\includegraphics[width=27em]{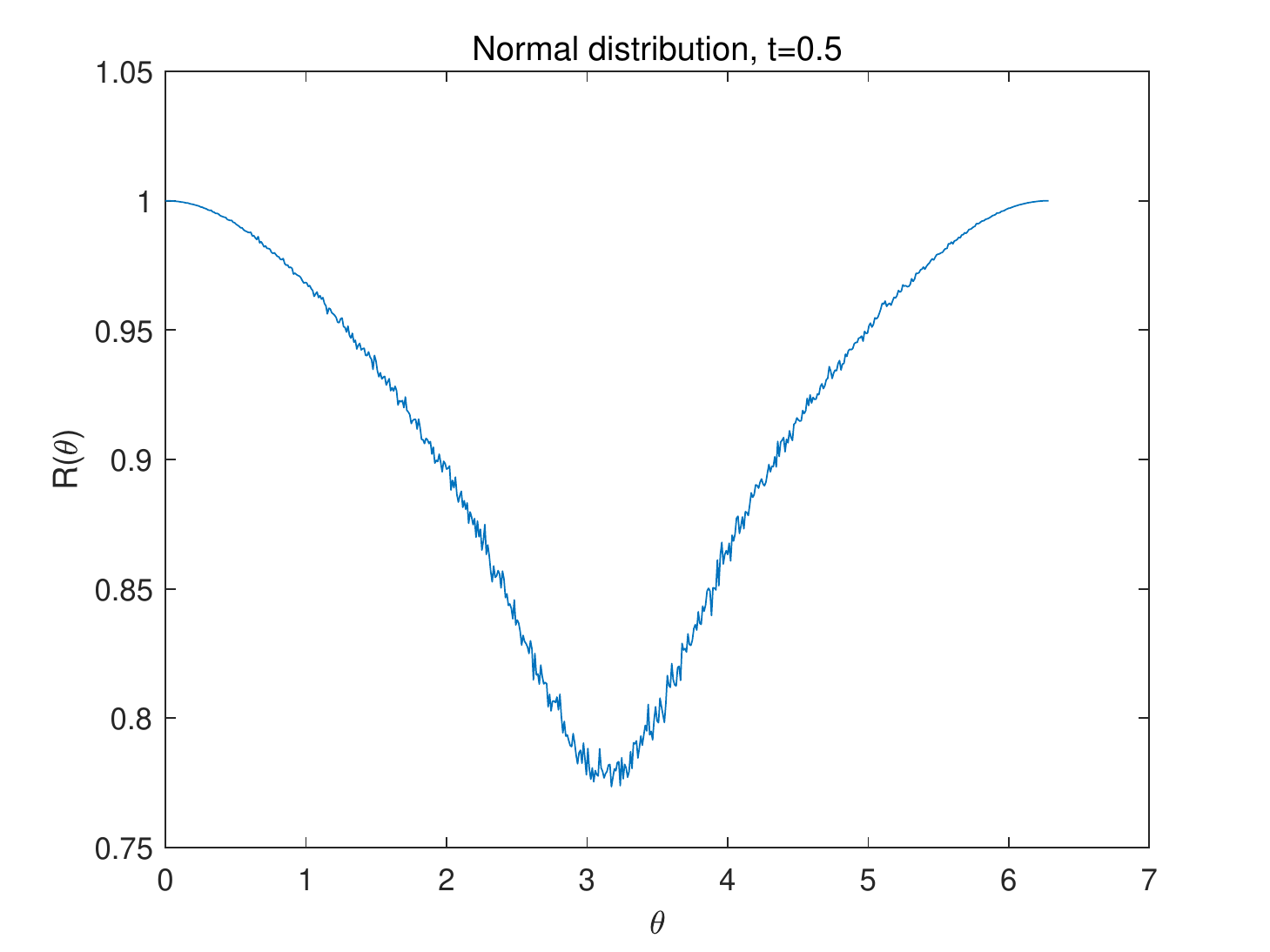}
	\end{center}
	\begin{center}
		Figure 6. $R(\theta, t)-\theta$ plot under covariates' distribution $\mathcal{N} (0, 1)$ and $t=0.5$.
	\end{center}
	\begin{center}
		\includegraphics[width=27em]{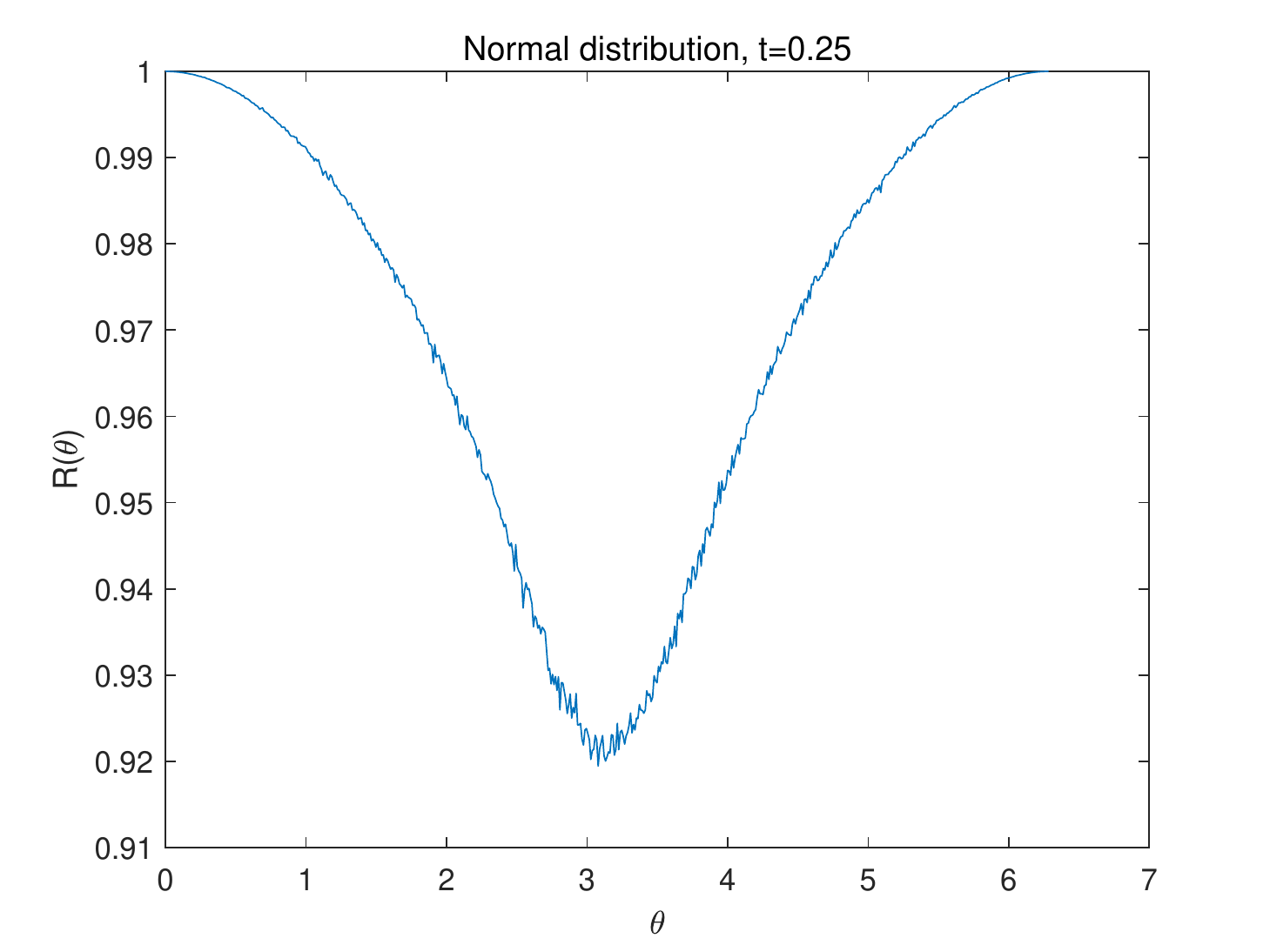}
	\end{center}
	\begin{center}
		Figure 7. $R(\theta, t)-\theta$ plot under covariates' distribution $\mathcal{N} (0, 1)$ and $t=0.25$.
	\end{center}
	As for the average reduction proportion $R(t)$, see the following tables.
	\begin{center}
		\begin{tabular}{|c|c|}
			\hline
			$t$ & $R(t) \ (\text{Distribution} \ \mathcal{U} \lbrack -1, 1 \rbrack)$ \\
			\hline
			$2$ & $0.8302$ \\
			\hline
			$1$ & $0.9022$ \\ 
			\hline
			$0.5$ & $0.9652$ \\
			\hline
		\end{tabular}
	\end{center}
	\begin{center}
		Table 1. The average reduction proportion under covariates' distribution $\mathcal{U} \lbrack -1, 1 \rbrack$.
	\end{center}
	\begin{center}
		\begin{tabular}{|c|c|}
			\hline
			$t$ & $R(t) \ (\text{Distribution} \ \mathcal{N} (0, 1))$ \\
			\hline
			$1$ & $0.8445$ \\
			\hline
			$0.5$ & $0.9141$ \\ 
			\hline
			$0.25$ & $0.9708$ \\
			\hline
		\end{tabular}
	\end{center}
	\begin{center}
		Table 2. The average reduction proportion under covariates' distribution $\mathcal{N} (0, 1)$.
	\end{center}
	From numerical results presented above, we have verified our previous assertion, i.e., that the linearly modified estimator $\tilde{\beta}_{lm}$ is nearly efficient in most general settings. In fact this will hold true as if
	\begin{equation}
	\frac{1}{p(X)} (\beta_T(X)-\beta_T) + \frac{1}{1-p(X)} (\beta_C(X)-\beta_C) \ \text{is strongly correlated with} \ X,       \tag*{}
	\end{equation}
	which is usually thought to be reasonable, if we incorporate sufficiently many outcome predictors in the covariates.
	
	\section{Discussion}
	\noindent We present a theoretical analysis for the effect of incorporating different kinds of covariates into the propensity score matching method on statistical efficiency of several estimators of average treatment effect in observational studies. We mainly focus on two types of covariate, i.e., the outcome predictors that are unassociated with treatment conditional on other covariates, and a specific kind of instrumental variables (exposure predictors) that are uncorrelated with potential outcomes when controlling for other covariates. The main result of this paper is that, including additional outcome predictors to the model helps lower asymptotic variance of efficient estimators, thus is informative for propensity score estimation. However, for those  estimators that use known propensity scores and are not asymptotically efficient, incorporating outcome predictors into the model do not improve their statistical efficiency, though do no harm to it as well. We also investigate the variance inflation effect caused by including a specific kind of instrumental variables aforementioned into propensity score estimation for all $\sqrt{n}$-consistent estimators reviewed in previous sections.
	
	Based on the fact that all known asymptotically efficient estimators of ATE require consistently estimating either the propensity score function or the conditional expectations of potential outcomes given the covariates under certain regularity conditions, we designed a linearly modified (LM) estimator based on the original IPW estimator. Our LM estimator can be viewed as a tradeoff between statistical efficiency and computational complexity, it enjoys the following two advantages comparing to the aforementioned estimators:
	\begin{enumerate}
	    \item On the one hand, the LM estimator is much more efficient than the original IPW estimator. In most general settings it's nearly efficient, this view point is numerically verified by simulation studies in Section 5, in which the potential outcomes satisfy a linear model and a logit model is assumed for the propensity score. A good result is that if the potential outcomes $Y^T$ and $Y^C$ are linearly dependent on covariates $X$ with same regression coefficient, then the LM estimator is fully efficient.
	    
	    \item On the other hand, when the propensity score function is known, the computation of LM estimator only relies on observed data and need not impose certain regularity conditions concerning the smoothness of propensity score function and boundedness of support of covariates. Unlike asymptotically efficient estimators aforementioned, the LM estimator is as easy to calculate as the original IPW estimator.
	\end{enumerate} 
	
	The discussion of this paper can be extended and continued along several directions. First, based on our previous results, effective covariate selection methods in causal inference must be proposed to exclude a specific kind of instrumental variables that are independent of potential outcomes conditional on other covariates and include outcome predictors. Second, when certain regularity conditions that enable us to consistently estimate the propensity scores are not satisfied, how can we estimate ATE as efficiently as we can? The LM estimator is an attempt but we will not stop here. Finally, in this paper we only show that additionally including a specific type of instrumental variables can lower statistical efficiency, since this issue has not been fully resolved now. Whether there exists a general criterion that differentiates instrumental variables harmful to efficient estimation from other covariates remains to be a problem for future researchers.

	\section{Appendix}
	\noindent In this appendix we give proofs to lemmas, theorems, and propositions presented but unproved in previous sections.
	\begin{proof}[Proof of Theorem 1]
		We first find out asymptotic expressions for empirical estimators $\hat{\mathbf{E}} \lbrack D_i Y_i \vert X_i \rbrack$ and $\hat{\mathbf{E}} \lbrack (1-D_i) Y_i \vert X_i \rbrack$. Take $\hat{\mathbf{E}} \lbrack D_i Y_i \vert X_i \rbrack$ as an example, for arbitrary $x$, we have
		\begin{align*}
		& \hat{\mathbf{E}} \lbrack D_i Y_i \vert X_i=x \rbrack - \mathbf{E} \lbrack D_i Y_i \vert X_i=x \rbrack  \\
		= & \frac{(1/n) \sum_{i=1}^{n} (D_i Y_i \mathbf{I}_{\lbrace X_i=x \rbrace} - \mathbf{E} \lbrack D_i Y_i \vert X_i=x \rbrack \mathbf{I}_{\lbrace X_i=x \rbrace})}{(1/n) \sum_{i=1}^{n} \mathbf{I}_{\lbrace X_i=x \rbrace}} \\
		= & \frac{1}{n \mathbf{P}(X_i=x)} \sum_{i=1}^{n} (D_i Y_i \mathbf{I}_{\lbrace X_i=x \rbrace} - \mathbf{E} \lbrack D_i Y_i \vert X_i=x \rbrack \mathbf{I}_{\lbrace X_i=x \rbrace}) + O_P(\frac{1}{n}),
		\end{align*}
		and similarly we can obtain that
		\begin{align*}
		& \hat{\mathbf{E}} \lbrack (1-D_i) Y_i \vert X_i=x \rbrack - \mathbf{E} \lbrack (1-D_i) Y_i \vert X_i=x \rbrack  \\
		= & \frac{1}{n \mathbf{P}(X_i=x)} \sum_{i=1}^{n} ((1-D_i) Y_i \mathbf{I}_{\lbrace X_i=x \rbrace} - \mathbf{E} \lbrack (1-D_i) Y_i \vert X_i=x \rbrack \mathbf{I}_{\lbrace X_i=x \rbrace}) + O_P(\frac{1}{n}). 
		\end{align*}
		Denote $n(x) = \sum_{i=1}^{n} \mathbf{I}_{\lbrack X_i=x \rbrack}$, now we have
		\begin{align*}
		& \tilde{\beta}_{imp} - \frac{1}{n} \sum_{i=1}^{n} (\beta_T(X_i) - \beta_C(X_i))   \\
		= & \frac{1}{n} \sum_{x} n(x) (\frac{\hat{\mathbf{E}} \lbrack D_i Y_i \vert X_i=x \rbrack - \mathbf{E} \lbrack D_i Y_i \vert X_i=x \rbrack}{p(x)} - \frac{\hat{\mathbf{E}} \lbrack (1-D_i) Y_i \vert X_i=x \rbrack - \mathbf{E} \lbrack (1-D_i) Y_i \vert X_i=x \rbrack}{1-p(x)})  \\
		= & \frac{1}{n} \sum_{x} \frac{n(x)}{n \mathbf{P}(X_i=x)} (\frac{\sum_{i=1}^{n} (D_i Y_i \mathbf{I}_{\lbrace X_i=x \rbrace} - \mathbf{E} \lbrack D_i Y_i \vert X_i=x \rbrack \mathbf{I}_{\lbrace X_i=x \rbrace})}{p(x)}    \\
		- & \frac{\sum_{i=1}^{n} ((1-D_i) Y_i \mathbf{I}_{\lbrace X_i=x \rbrace} - \mathbf{E} \lbrack (1-D_i) Y_i \vert X_i=x \rbrack \mathbf{I}_{\lbrace X_i=x \rbrace})}{1-p(x)}) + O_P(\frac{1}{n}) \\
		= & \frac{1}{n} \sum_{x} (\frac{\sum_{i=1}^{n} (D_i Y_i \mathbf{I}_{\lbrace X_i=x \rbrace} - \mathbf{E} \lbrack D_i Y_i \vert X_i=x \rbrack \mathbf{I}_{\lbrace X_i=x \rbrace})}{p(x)}    \\
		- & \frac{\sum_{i=1}^{n} ((1-D_i) Y_i \mathbf{I}_{\lbrace X_i=x \rbrace} - \mathbf{E} \lbrack (1-D_i) Y_i \vert X_i=x \rbrack \mathbf{I}_{\lbrace X_i=x \rbrace})}{1-p(x)}) + O_P(\frac{1}{n}) \\
		= & \frac{1}{n} \sum_{i=1}^{n} (\frac{D_i Y_i - \mathbf{E} \lbrack D_i Y_i \vert X_i \rbrack}{p(X_i)} - \frac{(1-D_i) Y_i - \mathbf{E} \lbrack (1-D_i) Y_i \vert X_i \rbrack} {1-p(X_i)}) + O_P(\frac{1}{n}).
		\end{align*}
		Therefore, $\tilde{\beta}_{imp}$ has influence function
		\begin{equation}
		\frac{DY}{p(X)} - \frac{(1-D)Y}{1-p(X)},  \tag*{}
		\end{equation}
		thus is not asymptotically efficient with asymptotic variance equal to
		\begin{equation}
		\mathbf{E} \lbrack \frac{\sigma_T^2(X) + \beta_T(X)^2}{p(X)} + \frac{\sigma_C^2(X) + \beta_C(X)^2}{1-p(X)} \rbrack - \beta^2,  \tag*{}
		\end{equation}
		which exceeds the semi-parametric efficiency bound by the magnitude of
		\begin{equation}
		\mathbf{E} \lbrack (\sqrt{\frac{1-p(X)}{p(X)}} \beta_T(X) + \sqrt{\frac{p(X)}{1-p(X)}} \beta_C(X))^2 \rbrack.  \tag*{}
		\end{equation}
		End of proof.
	\end{proof}
	
	\begin{proof}[Proof of Theorem 2]
		We directly compute influence function of $\tilde{\beta}_{ipw}$, in fact, it can be easily shown that
		\begin{align*}
		\tilde{\beta}_{ipw} = & \frac{(1/n) \sum_{i=1}^{n} D_i Y_i / p(X_i)}{(1/n) \sum_{i=1}^{n} D_i / p(X_i)} - \frac{(1/n) \sum_{i=1}^{n} (1 - D_i) Y_i / (1 - p(X_i))}{(1/n) \sum_{i=1}^{n} (1-D_i) / (1-p(X_i))}   \\
		= & \frac{1}{n} \sum_{i=1}^{n} \frac{D_i Y_i}{p(X_i)} (2 - \frac{1}{n} \sum_{i=1}^{n} \frac{D_i}{p(X_i)}) - \frac{1}{n} \sum_{i=1}^{n} \frac{(1 - D_i) Y_i}{1 - p(X_i)} (2 - \frac{1}{n} \sum_{i=1}^{n} \frac{1 - D_i}{1 - p(X_i)}) + O_P(\frac{1}{n}) \\
		= & \frac{1}{n} \sum_{i=1}^{n} \frac{D_i Y_i}{p(X_i)} - \frac{1}{n} \sum_{i=1}^{n} \frac{(1 - D_i) Y_i}{1 - p(X_i)} - \beta_T (\frac{1}{n} \sum_{i=1}^{n} \frac{D_i}{p(X_i)}-1) + \beta_C (\frac{1}{n} \sum_{i=1}^{n} \frac{1 - D_i}{1 - p(X_i)}-1) + O_P(\frac{1}{n}).
		\end{align*}
		Therefore $\tilde{\beta}_{ipw}$ has influence function equal to
		\begin{equation}
		\frac{DY}{p(X)} - \frac{(1-D)Y}{1-p(X)} - \beta_T \frac{D}{p(X)} + \beta_C \frac{1-D}{1-p(X)} + \beta.   \tag*{}
		\end{equation}
		Its asymptotic normal distribution can then be easily deduced.
	\end{proof}
	
	\begin{proof}[Proof of Lemma 3]
		We prove Lemma 3 by direct calculation. Appealing to the Law of Total Expectation, and the fact that $p(X)=p(X_0)$, we have
		\begin{align*}
		\mathbf{E} \lbrack \frac{\mathbf{E} \lbrack f \vert X \rbrack}{p(X)} + \frac{\mathbf{E} \lbrack g \vert X \rbrack}{1-p(X)} \rbrack = & \mathbf{E} \lbrack \frac{\mathbf{E} \lbrack f \vert X \rbrack}{p(X_0)} + \frac{\mathbf{E} \lbrack g \vert X \rbrack}{1-p(X_0)} \rbrack   \\
		= & \mathbf{E} \lbrack \mathbf{E} \lbrack \frac{\mathbf{E} \lbrack f \vert X \rbrack}{p(X_0)} + \frac{\mathbf{E} \lbrack g \vert X \rbrack}{1-p(X_0)} \vert X_0 \rbrack \rbrack \\
		= & \mathbf{E} \lbrack \frac{\mathbf{E} \lbrack \mathbf{E} \lbrack f \vert X \rbrack \vert X_0 \rbrack}{p(X_0)} + \frac{\mathbf{E} \lbrack \mathbf{E} \lbrack g \vert X \rbrack \vert X_0 \rbrack}{1-p(X_0)} \rbrack  \\
		= & \mathbf{E} \lbrack \frac{\mathbf{E} \lbrack f \vert X_0 \rbrack}{p(X_0)} + \frac{\mathbf{E} \lbrack g \vert X_0 \rbrack}{1-p(X_0)} \rbrack
		\end{align*}
		End of proof.
	\end{proof}
	
	\begin{proof}[Proof of Lemma 4]
		Based on the fact that $p(X)=p(X_0)$, using the Law of total expectation and Jensen's inequality we find that
		\begin{align*}
		& \mathbf{E} \lbrack (\sqrt{\frac{1-p(X)}{p(X)}} \mathbf{E} \lbrack f \vert X \rbrack + \sqrt{\frac{p(X)}{1-p(X)}} \mathbf{E} \lbrack g \vert X \rbrack)^2 \rbrack
		= \mathbf{E} \lbrack \mathbf{E} \lbrack (\sqrt{\frac{1-p(X_0)}{p(X_0)}} \mathbf{E} \lbrack f \vert X \rbrack + \sqrt{\frac{p(X_0)}{1-p(X_0)}} \mathbf{E} \lbrack g \vert X \rbrack)^2 \vert X_0 \rbrack \rbrack  \\
		\geq & \mathbf{E} \lbrack (\mathbf{E} \lbrack \sqrt{\frac{1-p(X_0)}{p(X_0)}} \mathbf{E} \lbrack f \vert X \rbrack + \sqrt{\frac{p(X_0)}{1-p(X_0)}} \mathbf{E} \lbrack g \vert X \rbrack \vert X_0 \rbrack)^2 \rbrack 
		\geq \mathbf{E} \lbrack (\sqrt{\frac{1-p(X_0)}{p(X_0)}} \mathbf{E} \lbrack f \vert X_0 \rbrack + \sqrt{\frac{p(X_0)}{1-p(X_0)}} \mathbf{E} \lbrack g \vert X_0 \rbrack)^2 \rbrack.
		\end{align*}
		End of proof.
	\end{proof}
	
	\begin{proof}[Proof of Theorem 6]
		Since $p(Y^T, Y^C \vert X) = p(Y^T, Y^C \vert X_1)$, we know that
		\begin{equation}
		\sigma_T^2(X) = \sigma_T^2(X_1), \ \beta_T(X) = \beta_T(X_1), \ \sigma_C^2(X) = \sigma_C^2(X_1), \ \beta_C(X) = \beta_C(X_1).  \tag*{}
		\end{equation}
		Hence actually we only need to prove that for any non-negative measurable functions $f(X)$ and $g(X)$, the following inequality holds:
		\begin{equation}
		\mathbf{E} \lbrack \frac{f(X_1)}{p(X_1)} + \frac{g(X_1)}{1-p(X_1)} \rbrack \leq \mathbf{E} \lbrack \frac{f(X_1)}{p(X)} + \frac{g(X_1)}{1-p(X)} \rbrack.     \tag*{}
		\end{equation}
		Appealing to Jensen's inequality, note that both $1/x$ and $1/(1-x)$ are convex in $(0, 1)$, we have
		\begin{align*}
		\mathbf{E} \lbrack \frac{f(X_1)}{p(X)} + \frac{g(X_1)}{1-p(X)} \rbrack = & \mathbf{E} \lbrack f(X_1) \mathbf{E} \lbrack \frac{1}{p(X)} \vert X_1 \rbrack + g(X_1) \mathbf{E} \lbrack \frac{1}{1-p(X)} \vert X_1 \rbrack \rbrack   \\
		\geq & \mathbf{E} \lbrack f(X_1) \frac{1}{\mathbf{E} \lbrack p(X) \vert X_1 \rbrack} + g(X_1) \frac{1}{\mathbf{E} \lbrack 1 - p(X) \vert X_1 \rbrack} \rbrack.
		\end{align*}
		Since
		\begin{equation}
		\mathbf{E} \lbrack p(X) \vert X_1 \rbrack = \mathbf{E} \lbrack \mathbf{E} \lbrack D \vert X \rbrack \vert X_1 \rbrack = \mathbf{E} \lbrack D \vert X_1 \rbrack = p(X_1),    \tag*{}
		\end{equation}
		the desired inequality is shown, end of proof.
	\end{proof}
	
	\begin{proof}[Proof of Lemma 6]
		Since influence function of $\tilde{\beta}_{ipw}$ has been obtained in proof of theorem 2, we only need to find out influence function of $\tilde{x}_{ipw}$, then the corresponding asymptotic covariance matrices can be directly calculated. In fact, we have
		\begin{align*}
		\tilde{x}_{ipw} = & \frac{(1/n) \sum_{i=1}^{n} D_i X_i/p(X_i)}{(1/n) \sum_{i=1}^{n} D_i/p(X_i)} - \frac{(1/n) \sum_{i=1}^{n} (1 - D_i) X_i/(1 - p(X_i))}{(1/n) \sum_{i=1}^{n} (1-D_i)/(1-p(X_i))} \\
		= & \frac{1}{n} \sum_{i=1}^{n} \frac{D_i X_i}{p(X_i)} (2 - \frac{1}{n} \sum_{i=1}^{n} \frac{D_i}{p(X_i)}) - \frac{1}{n} \sum_{i=1}^{n} \frac{(1 - D_i) X_i}{1 - p(X_i)} (2 - \frac{1}{n} \sum_{i=1}^{n} \frac{1 - D_i}{1 - p(X_i)}) + O_P(\frac{1}{n}) \\
		= & \frac{1}{n} \sum_{i=1}^{n} \frac{D_i X_i}{p(X_i)} - \frac{1}{n} \sum_{i=1}^{n} \frac{(1 - D_i) X_i}{1 - p(X_i)} - \mathbf{E} X (\frac{1}{n} \sum_{i=1}^{n} \frac{D_i}{p(X_i)} - \frac{1}{n} \sum_{i=1}^{n} \frac{1 - D_i}{1 - p(X_i)}) + O_P(\frac{1}{n}).
		\end{align*}
		Therefore $\tilde{x}_{ipw}$ has influence function equal to
		\begin{equation}
		\frac{DX}{p(X)} - \frac{(1-D)X}{1-p(X)} - \mathbf{E} X (\frac{D}{p(X)} - \frac{1-D}{1-p(X)}).   \tag*{}
		\end{equation}
		We then can easily verify expressions of $\text{asycov} (\tilde{x}_{ipw}, \tilde{x}_{ipw})$ and $\text{asycov} (\tilde{x}_{ipw}, \tilde{\beta}_{ipw})$ given in Lemma 3.
	\end{proof}
	
	\begin{proof}[Proof of Proposition 1]
		We only need to notice that since
		\begin{equation}
		\tilde{x}_{ipw} = \frac{1}{n} \sum_{i=1}^{n} (\frac{D_i X_i}{p(X_i)} - \frac{(1 - D_i) X_i}{1 - p(X_i)} - \mathbf{E} X (\frac{D_i}{p(X_i)} - \frac{1 - D_i}{1 - p(X_i)})) + O_P(\frac{1}{n}),   \tag*{}
		\end{equation}
		and the fact that
		\begin{equation}
		\mathbf{E} \lbrack \frac{DX}{p(X)} - \frac{(1-D)X}{1-p(X)} - \mathbf{E} X (\frac{D}{p(X)} - \frac{1-D}{1-p(X)}) \rbrack = 0. \tag*{}
		\end{equation}
		Then $\tilde{x}_{ipw} = O_P(1 / \sqrt{n})$, therefore if $\hat{\alpha} - \alpha = O_P(1 / \sqrt{n})$, we have
		\begin{equation}
		\sqrt{n} (\tilde{\beta}_{ipw} - \hat{\alpha}^T \tilde{x}_{ipw} - \beta) = \sqrt{n} (\tilde{\beta}_{ipw} - \alpha^T \tilde{x}_{ipw} - \beta) + O_P(\frac{1}{\sqrt{n}}),  \tag*{}
		\end{equation}
		Proposition 1 then naturally follows.
	\end{proof}
	
	\begin{proof}[Proof of Proposition 2]
		First note that $\bar{X} - \mathbf{E} X = O_P(1/n)$, therefore
		\begin{align*}
		\widehat{\text{asycov}} (\tilde{x}_{ipw}, \tilde{x}_{ipw}) = & \frac{1}{n} \sum_{i=1}^{n} \frac{1}{p(X_i) (1-p(X_i))} (X_i - \bar{X}) (X_i - \bar{X})^T  \\
		= & \frac{1}{n} \sum_{i=1}^{n} \frac{1}{p(X_i) (1-p(X_i))} (X_i - \mathbf{E} X) (X_i - \mathbf{E} X)^T + O_P(\frac{1}{\sqrt{n}})  \\
		= & \text{asycov} (\tilde{x}_{ipw}, \tilde{x}_{ipw}) + O_P(\frac{1}{\sqrt{n}}).
		\end{align*}
		Second, by definition we have
		\begin{align*}
		\widehat{\text{asycov}} (\tilde{x}_{ipw}, \tilde{\beta}_{ipw}) = & \frac{1}{n} \sum_{j=1}^{n} (X_j - \bar{X}) (\frac{D_j}{p(X_j)^2} (Y_j^T - \frac{1}{n} \sum_{i=1}^{n} \frac{D_i Y_i}{p(X_i)}) + \frac{1-D_j}{(1-p(X_j))^2} (Y_j^C - \frac{1}{n} \sum_{i=1}^{n} \frac{(1-D_i) Y_i}{1-p(X_i)}))   \\
		= & \frac{1}{n} \sum_{j=1}^{n} (X_j - \mathbf{E}X) (\frac{D_j}{p(X_j)^2} (Y_j^T - \beta_T) + \frac{1-D_j}{(1-p(X_j))^2} (Y_j^C - \beta_C)) + O_P(\frac{1}{\sqrt{n}})  \\
		= & \mathbf{E} \lbrack (X - \mathbf{E}X) (\frac{D}{p(X)^2} (Y^T - \beta_T) + \frac{1-D}{(1-p(X))^2} (Y^C - \beta_C)) \rbrack + O_P(\frac{1}{\sqrt{n}})  \\
		= & \mathbf{E} \lbrack (X - \mathbf{E}X) (\frac{1}{p(X)} (\beta_T(X) - \beta_T) + \frac{1}{1-p(X)} (\beta_C(X) - \beta_C)) \rbrack + O_P(\frac{1}{\sqrt{n}})  \\
		= & \text{asycov} (\tilde{x}_{ipw}, \tilde{\beta}_{ipw}) + O_P(\frac{1}{\sqrt{n}}).
		\end{align*}
		Hence both $\widehat{\text{asycov}} (\tilde{x}_{ipw}, \tilde{x}_{ipw})$ and $\widehat{\text{asycov}} (\tilde{x}_{ipw}, \tilde{\beta}_{ipw})$ are $\sqrt{n}-\text{consistent}$ estimators, end of proof.
	\end{proof}
	
	\begin{proof}[Proof of Theorem 7]
		We would like to take
		\begin{equation}
		\hat{\alpha} = \widehat{\text{asycov}} (\tilde{x}_{ipw}, \tilde{x}_{ipw})^{-1} \widehat{\text{asycov}} (\tilde{x}_{ipw}, \tilde{\beta}_{ipw}) \ \text{and} \ \alpha = \text{asycov} (\tilde{x}_{ipw}, \tilde{x}_{ipw})^{-1} \text{asycov} (\tilde{x}_{ipw}, \tilde{\beta}_{ipw}),  \tag*{}
		\end{equation}
		then we know that
		\begin{equation}
		\text{asyvar}(\tilde{\beta}_{lm}) = \text{asyvar}(\tilde{\beta}_{ipw}) - \text{asycov} (\tilde{x}_{ipw}, \tilde{\beta}_{ipw})^T \text{asycov} (\tilde{x}_{ipw}, \tilde{x}_{ipw})^{-1} \text{asycov} (\tilde{x}_{ipw}, \tilde{\beta}_{ipw}).   \tag*{}
		\end{equation}
		Then $\text{asyvar}(\tilde{\beta}_{lm}) = \text{asyvar}(\tilde{\beta}_{ipw})$ if and only if
		\begin{equation}
		\text{asycov} (\tilde{x}_{ipw}, \tilde{\beta}_{ipw}) = 0,   \tag*{}
		\end{equation}
		which is equivalent to the fact that
		\begin{equation}
		\frac{1}{p(X)} (\beta_T(X)-\beta_T) + \frac{1}{1-p(X)} (\beta_C(X)-\beta_C) \ \text{is uncorrelated with} \ X.       \tag*{}
		\end{equation}
		End of proof.
	\end{proof}
	
	\begin{proof}[Proof of Theorem 8]
		On the one hand, now it's easy to figure out that
		\begin{equation}
		\text{asycov} (\tilde{x}_{ipw}, \tilde{\beta}_{ipw}) = \mathbf{E} \lbrack \frac{1}{p(X)(1-p(X))} (X-\mathbf{E}X) (X-\mathbf{E}X)^T \rbrack c,   \tag*{}
		\end{equation}
		therefore the efficiency gain of $\tilde{\beta}_{lm}$ from $\tilde{\beta}_{ipw}$ is equal to
		\begin{equation}
		c^T \mathbf{E} \lbrack \frac{1}{p(X)(1-p(X))} (X-\mathbf{E}X) (X-\mathbf{E}X)^T \rbrack c.   \tag*{}
		\end{equation}
		On the other hand, we know that the difference between $\text{asyvar}(\tilde{\beta}_{ipw})$ and the semi-parametric efficiency bound equals to
		\begin{align*}
		& \mathbf{E} \lbrack (\sqrt{\frac{1-p(X)}{p(X)}}(\beta_T(X)-\beta_T) + \sqrt{\frac{p(X)}{1-p(X)}}(\beta_C(X)-\beta_C))^2 \rbrack  \\
		= & \mathbf{E} \lbrack (\sqrt{\frac{1-p(X)}{p(X)}} c^T (X - \mathbf{E} X) + \sqrt{\frac{p(X)}{1-p(X)}} c^T (X - \mathbf{E} X))^2 \rbrack   \\
		= & \mathbf{E} \lbrack \frac{1}{p(X) (1-p(X))}(c^T (X - \mathbf{E} X))^2 \rbrack = c^T \mathbf{E} \lbrack \frac{1}{p(X)(1-p(X))} (X-\mathbf{E}X) (X-\mathbf{E}X)^T \rbrack c.
		\end{align*}
		Therefore, $\text{asyvar}(\tilde{\beta}_{lm})$ reaches the semi-parametric efficiency bound.
	\end{proof}
	
	\bibliographystyle{plain}
	\bibliography{MyBib}
	
\end{document}